\documentclass[journal,10pt,a4paper]{IEEEtran}

\usepackage{microtype}                
\usepackage{newtxtext, newtxmath}     
\usepackage{amsmath} 
\usepackage{stmaryrd}                 
\usepackage{xcolor}                   
\usepackage{ulem}                     
\usepackage{booktabs}                 
\usepackage{threeparttable}           
\usepackage{multirow}                 
\usepackage{makecell}                 
\usepackage{pifont}                   
\newcommand{\cmark}{\ding{51}}        
\newcommand{\xmark}{\ding{55}}        

\usepackage{graphicx}                 
\usepackage{wrapfig}                  
\usepackage{float}                    
\usepackage[caption=false,font=normalsize,labelfont=rm,textfont=rm]{subfig}
\usepackage{chngpage}                 
\usepackage{lipsum}                   

\usepackage{algorithm}                
\usepackage{algorithmic}              

\usepackage[lambda,advantage,operators,sets,adversary,landau,probability,
            notions,logic,ff,mm,primitives,events,complexity,asymptotics,keys]{cryptocode}

\newtheorem{definition}{Definition}[section]  
\newtheorem{theorem}{Theorem}[section]
\newtheorem{assumption}{Assumption}[section]
\newenvironment{proof}{\par\noindent\textit{Proof. }\ignorespaces}{\hfill$\square$\par}

\usepackage{hyperref}
\hypersetup{
    colorlinks=true,
    linkcolor=blue,
    citecolor=blue,
    urlcolor=blue,
    breaklinks=true
}

\begin{document}

\title{\textit{SRFed}: Mitigating Poisoning Attacks in Privacy-Preserving Federated Learning \\
with Heterogeneous Data}

\author{Yiwen Lu\IEEEauthorrefmark{1},%
\IEEEauthorrefmark{2}Corresponding Author\\
\IEEEauthorrefmark{1}School of Mathematics, Nanjing University, Nanjing 210093, China\\
\IEEEauthorrefmark{2}E-mail: luyw@smail.nju.edu.cn
}
\maketitle
\begin{abstract}
Federated Learning (FL) enables collaborative model training without exposing clients' private data, and has been widely adopted in privacy-sensitive scenarios. However, FL faces two critical security threats: curious servers that may launch inference attacks to reconstruct clients' private data, and compromised clients that can launch poisoning attacks to disrupt model aggregation. Existing solutions mitigate these attacks by combining mainstream privacy-preserving techniques with defensive aggregation strategies. However, they either incur high computation and communication overhead or perform poorly under non-independent and identically distributed (Non-IID) data settings. To tackle these challenges, we propose SRFed, an efficient Byzantine-robust and privacy-preserving FL framework for Non-IID scenarios. First, we design a decentralized efficient functional encryption (DEFE) scheme to support efficient model encryption and non-interactive decryption. DEFE also eliminates third-party reliance and defends against server-side inference attacks. Second, we develop a privacy-preserving defensive model aggregation mechanism based on DEFE. This mechanism filters poisonous models under Non-IID data by layer-wise projection and clustering-based analysis. Theoretical analysis and extensive experiments show that SRFed outperforms state-of-the-art baselines in privacy protection, Byzantine robustness, and efficiency.
\end{abstract}



\section{Introduction}

Federated learning (FL) has emerged as a promising paradigm for distributed machine learning, which enables multiple clients to collaboratively train a global model without sharing their private data. In a typical FL setup, multiple clients periodically train local models using their private data and upload model updates to a central server, which aggregates these updates to obtain a global model with enhanced performance. Due to its ability to protect data privacy, FL has been widely applied in real-world scenarios, such as intelligent driving \cite{Nc_driving, FL_intelligent_driving1, FL_intelligent_driving2, FL_intelligent_driving3, Nc_driving2}, medical diagnosis \cite{TIFS_medical1, FL_medical_diagnosis1, FL_medical_diagnosis2, FL_medical_diagnosis3}, and intelligent recommendation systems \cite{FL_recommendation1, TIFS_recommendation2, TIFS_recommendation3, TIFS_recommendation1}.

Although FL avoids direct data exposure, it is not immune to privacy and security risks. Existing studies \cite{TIFS_inference_attack1,   model_inversion_attacks1, TIFS_inference_attack2} have shown that the curious server may launch inference attacks to reconstruct sensitive data samples of clients from the model updates. This may lead to the leakage of clients' sensitive information, e.g., medical diagnosis records, which can be exploited by adversaries for launching further malicious activities. Moreover, FL is also vulnerable to poisoning attacks \cite{TIFS_poisoning_attack2, poisoning_step, TIFS_poisoning_attacks}. The adversaries may manipulate some clients to execute malicious local training and upload poisonous model updates to mislead the global model \cite{poisoning_step, poisoning_attacks, Nc_poisoning_attack}. This attack will damage the performance of the global model and lead to incorrect decisions in downstream tasks such as medical diagnosis and intelligent driving.

To address privacy leakage issues, existing privacy-preserving federated learning (PPFL) methods primarily rely on techniques such as differential privacy (DP) \cite{DP_2constants, DP_Expert_Systems_with_Applications, DP_adaptive, TIFS_DP1, TIFS_DP2}, secure multi-party computation (SMC) \cite{SMC_Median_LSFL, SMC_information_sciences}, and homomorphic encryption (HE) \cite{HE_shieldfl, HE_ASPLOS25}. However, DP typically leads to a reduction in global model accuracy, while SMC and HE incur substantial computational and communication overheads in FL. Recently, lightweight functional encryption (FE) schemes \cite{NDDFE_2022_Big_data, NIFE_2024_BSR_FL} have been applied in FL. FE enables the server to aggregate encrypted models and directly obtain the decrypted results via a functional key, which avoids the accuracy loss in DP and the extra communication overhead in HE/SMC. However, existing FE schemes rely on third parties for functional key generation, which introduces potential privacy risks.


To mitigate poisoning attacks, existing Byzantine-robust FL methods \cite{foolsgold, Krum, Median_original, Nc_tongji} typically adopt defensive aggregation strategies, which filter out malicious model updates based on statistical distances or performance-based criteria. However, these strategies rely on the assumption that clients' data distributions are homogeneous, leading to poor performance in non-independent and identically distributed (Non-IID) data settings. Moreover, they require access to plaintext model updates, which directly contradicts the design goal of PPFL \cite{NIFE_2024_BSR_FL}. Recently, privacy-preserving Byzantine-robust FL methods \cite{pbfl, HE_shieldfl} have been proposed to address both privacy and poisoning attacks. However, these methods still suffer from limitations such as accuracy loss, excessive overhead, and limited effectiveness in Non-IID environments, as they merely combine the PPFL with existing defensive aggregation strategies. As a result, there is still a lack of practical solutions that can simultaneously ensure privacy protection and Byzantine robustness in Non-IID data scenarios.

To address the limitations of existing FL solutions, we propose a novel Byzantine-robust privacy-preserving FL method, SRFed. SRFed achieves efficient privacy protection and Byzantine robustness in Non-IID data scenarios through two key designs. First, we propose a new functional encryption scheme, DEFE, to protect clients’ model privacy and resist inference attacks from the server. Compared with existing FE schemes, DEFE eliminates reliance on third parties through distributed key generation and improves decryption efficiency by reconstructing the ciphertext. Second, we develop a privacy-preserving robust aggregation strategy based on secure layer-wise projection and clustering. This strategy resists the poisoning attacks in Non-IID data scenarios. Specifically, this strategy first decomposes client models layer by layer, projects each layer onto the corresponding layer of the global model. Then, it performs clustering analysis on the projection vectors to filter malicious updates and aggregates the remaining benign models. DEFE supports the above secure layer-wise projection computation and enables privacy-preserving model aggregation. Finally, we evaluate SRFed on multiple datasets with varying levels of heterogeneity. Theoretical analysis and experimental results demonstrate that SRFed achieves strong privacy protection, Byzantine robustness, and high efficiency. In summary, the core contributions of this paper are as follows. 
\begin{itemize}
    \item We propose a novel secure and robust FL method, SRFed, which simultaneously guarantees privacy protection and Byzantine robustness in Non-IID data scenarios.
    \item We design an efficient functional encryption scheme, which not only effectively protects the local model privacy but also enables efficient and secure model aggregation.
    \item We develop a privacy-preserving robust aggregation strategy, which effectively defends against poisoning attacks in Non-IID scenarios and generates high-quality aggregated models.
    \item We implement the prototype of SRFed and validate its performance in terms of privacy preservation, Byzantine robustness, and computational efficiency. Experimental results show that SRFed outperforms state-of-the-art baselines in all aspects.
\end{itemize}

\section{Related Works}
\subsection{Privacy-Preserving Federated Learning}

To safeguard user privacy, current research on privacy-preserving federated learning (PPFL) mainly focuses on protecting gradient information. Existing solutions are primarily built upon four core technologies: Differential Privacy (DP) \cite{DP_2constants, DP_Expert_Systems_with_Applications,DP_adaptive}, Secure Multi-Party Computation (SMC) \cite{SMC_Median_LSFL, SMC_information_sciences}, Homomorphic Encryption (HE) \cite{HE_shieldfl, TIFS_homomorphic_encryption,  HE_ASPLOS25}, and Functional Encryption (FE) \cite{NDDFE_2022_Big_data, NIFE_2024_BSR_FL, ESFL_ZHH_2025, FE_TIFS}. DP achieves data indistinguishability by injecting calibrated noise into raw data, thus ensuring privacy with low computational overhead. Miao et al. \cite{DP_adaptive} proposed a DP-based ESFL framework that adopts adaptive local DP to protect data privacy. However, the injected noise inevitably degrades the model accuracy. To avoid accuracy loss, SMC and HE employ cryptographic primitives to achieve privacy preservation. SMC enables distributed aggregation while keeping local gradients confidential, revealing only the aggregated model update. Zhang et al. \cite{SMC_Median_LSFL} introduced LSFL, a secure FL framework that applies secret sharing to split and transmit local parameters to two non-colluding servers for privacy-preserving aggregation. HE allows direct computation on encrypted data and produces decrypted results identical to plaintext computations. This property preserves privacy without sacrificing accuracy. Ma et al. \cite{HE_shieldfl} developed ShieldFL, a robust FL framework based on two-trapdoor HE, which encrypts all local gradients and achieves aggregation of encrypted gradients. Despite their strong privacy guarantees, SMC/HE-based FL methods incur substantial computation and communication overhead, posing challenges for large-scale deployment. To address these issues, FE has been introduced into FL. FE avoids noise injection and eliminates the high overhead caused by multi-round interactions or complex homomorphic operations. Chen et al. \cite{NDDFE_2022_Big_data} proposed ESB-FL, an efficient secure FL framework based on non-interactive designated decrypter FE (NDD-FE), which protects local data privacy but relies on a trusted third-party entity. Yu et al. \cite{FE_TIFS} further proposed PrivLDFL, which employs a dynamic decentralized multi-client FE (DDMCFE) scheme to preserve privacy in decentralized settings. However, both FE-based methods require discrete logarithm-based decryption, which is typically a time-consuming operation. To overcome these limitations, we propose a decentralized efficient functional encryption (DEFE) scheme that achieves privacy protection and high computational and communication efficiency.

\begin{table}[t] 
\centering
\caption{COMPARISON BETWEEN OUR METHOD WITH PREVIOUS WORK}
\label{tab:comparison}
\resizebox{\columnwidth}{!}{
\begin{threeparttable}
    \begin{tabular}{|c|c|c|c|c|c|}  
      \hline  
      \textbf{Methods} & \textbf{Privacy \ Protection} & \textbf{Defense Mechanism} & \textbf{Efficient} & \textbf{Non-IID} & \textbf{Fidelity} \\ \hline
      ESFL \cite{DP_adaptive}         & Local DP                     & Local DP                   & \cmark            & \xmark & \xmark \\ \hline
      PBFL \cite{pbfl}         & CKKS                         & Cosine similarity          & \xmark            & \xmark  & \cmark \\ \hline
      ESB-FL \cite{NDDFE_2022_Big_data}      & NDD-FE                       & \xmark                     & \cmark            & \xmark  & \cmark \\ \hline
      Median \cite{Median_original}      & \xmark                       & Median                     & \cmark            & \xmark     & \cmark \\ \hline
      FoolsGold \cite{foolsgold}       & \xmark                       & Cosine similarity          & \cmark            & \xmark      & \cmark \\ \hline
      ShieldFL \cite{HE_shieldfl}     & HE                            & Cosine similarity          & \xmark            & \xmark & \cmark \\ \hline
      PrivLDFL \cite{FE_TIFS}        & DDMCFE                       & \xmark          & \cmark            & \xmark    & \cmark \\ \hline
      Biscotti \cite{Mkrum}    & DP                            & Euclidean distance         & \cmark            & \xmark  & \xmark \\ \hline
      \textbf{SRFed}  & \textbf{DEFE} & \makecell{\textbf{Layer-wise projection} \\ \textbf{and clustering}} & \textbf{\cmark} & \textbf{\cmark} & \textbf{\cmark} \\ \hline  
    \end{tabular}

    \begin{tablenotes}
      \large
      \item \textbf{Notes}: The symbol "\cmark" indicates that it owns this property; "\xmark" indicates that it does not own this property. "Fidelity" indicates that the method has no accuracy loss when there is no attack. "Non-IID" indicates that the method is Byzantine-Robust under Non-IID data environments.
    \end{tablenotes}
  \end{threeparttable}
}
\end{table}

\subsection{Privacy-Preserving Federated Learning Against Malicious Participants}

To resist poisoning attacks, several defensive aggregation rules have been proposed in FL. FoolsGold \cite{foolsgold}, proposed by Fung et al. \cite{foolsgold}, reweights clients' contributions by computing the cosine similarity of their historical gradients. Krum \cite{Krum} selects a single client update that is closest, in terms of Euclidean distance, to the majority of other updates in each iteration. Median \cite{Median_original} mitigates the effect of malicious clients by taking the median value of each model parameter across all clients. However, the above aggregation rules require access to plaintext model updates. This makes them unsuitable for direct application in PPFL. To achieve Byzantine-robust PPFL, Shiyan et al. \cite{Mkrum} proposed Biscotti. Biscotti leverages DP to protect local gradients while using the Krum algorithm to mitigate poisoning attacks. Nevertheless, the injected noise in DP reduces the accuracy of the aggregated model. To overcome this limitation, Zhang et al. \cite{SMC_Median_LSFL} proposed LSFL. LSFL employs SMC to preserve privacy and uses Median-based aggregation for poisoning defense. However, its dual-server architecture introduces significant communication overhead. In addition, Ma et al. \cite{HE_shieldfl} and Miao et al. \cite{pbfl} proposed ShieldFL and PBFL, respectively. Both schemes adopt HE to protect local gradients and cosine similarity to defend against poisoning attacks. However, they suffer from high computational complexity and limited robustness under non-IID data settings. To address these challenges, we propose a novel Byzantine-robust and privacy-preserving federated learning method. Table \ref{tab:comparison} compares previous schemes with our method.

\section{Problem Statement}
\subsection{System Model}
The system model of SRFed comprises two roles: the aggregation server and clients.
\begin{itemize}
    \item \textbf{Clients:} Clients are nodes with limited computing power and heterogeneous data. In real-world scenarios, data heterogeneity typically arises across clients (e.g., intelligent vehicles) due to differences in usage patterns, such as driving habits. Each client is responsible for training its local model based on its own data. To protect data privacy, the models are encrypted and submitted to the server for aggregation.

    \item \textbf{Server:} The server is a node with strong computing power (e.g., service provider of intelligent vehicles). It collects encrypted local models from clients, conducts model detection, and then aggregates selected models and distributes the aggregated model back to clients for the next training round.
    
\end{itemize}

\subsection{Threat Model}
We consider the following threat model:

\textit{1) Honest-But-Curious server:} The server honestly follows the FL protocol but attempts to infer clients’ private data. Specifically, upon receiving encrypted local models from the clients, the server may launch inference attacks on the encrypted models and exploit intermediate computational results (e.g., layer-wise projections and aggregated outputs) to extract sensitive information of clients.

\textit{2) Malicious clients:} We consider a FL scenario where a certain proportion of clients are malicious. These malicious clients conduct model poisoning attacks to poison the global model, thereby disrupting the training process. Specifically, we focus on the following attack types:
\begin{itemize}
    \item Targeted poisoning attack. This attack aims to poison the global model so that it incurs erroneous predictions for the samples of a specific label. More specifically, we consider the prevalent label-flipping attack \cite{HE_shieldfl}. Malicious clients remap samples labeled $l_{src}$ to a chosen target label $l_{tar}$ to obtain a poisonous dataset $D_i^*$. Subsequently, they train local models based on $D_i^*$ and submit the poisonous models to the server for aggregation. As a result, the global model is compromised, leading to misclassification of source-label samples as the target label during inference.
    
    \item Untargeted poisoning attack. This attack aims to degrade the global model’s performance on the test samples of all classes. Specifically, we consider the classic Gaussian Attack \cite{SMC_Median_LSFL}. The malicious clients first train local models based on the clean dataset. Then, they inject Gaussian noise into the model parameters and submit the malicious models to the server. Consequently, the aggregated model exhibits low accuracy across test samples of all classes.
\end{itemize}

\subsection{Design Goals}
Under the defined threat model, SRFed aims to ensure the following security and performance guarantees:
\begin{itemize}
    \item \textbf{Confidentiality.} SRFed should ensure that any unauthorized entities (e.g., the server) cannot infer clients’ private training data from the encrypted models or intermediate results. 
    \item \textbf{Robustness.} SRFed should mitigate poisoning attacks launched by malicious clients under Non-IID data settings while maintaining the quality of the final aggregated model.
    \item \textbf{Efficiency.} SRFed should ensure efficient FL, with the introduced DEFE scheme and robust aggregation strategy incurring only limited computation and communication overhead. 
\end{itemize}

\section{Building Blocks}
\subsection{NDD-FE Scheme}
NDD-FE \cite{NDDFE_2022_Big_data} is a functional encryption scheme that supports the inner-product computation between a private vector $\boldsymbol{x}$ and a public vector $\boldsymbol{y}$. NDD-FE involves three roles, i.e., generator, encryptor, and decryptor, to elaborate on its construction. 
\begin{itemize}
    \item \textbf{NDD-FE.Setup(${1}^\lambda$)$\rightarrow$$pp$}: It is executed by the generator. It takes the security parameter ${1}^\lambda$ as input and generates the system public parameters $pp=(G, p, g)$ and a secure hash function $H_1$.
    
    \item \textbf{NDD-FE.KeyGen($pp$)$\rightarrow$$(pk,sk)$}: It is executed by all roles. It takes $pp$ as input and outputs public/secret keys $(pk,sk)$. Let $(pk_1, sk_1),$ $(pk_{2i}, sk_{2i})$ and $(pk_3, sk_3)$ denote the public/secret key pairs of the generator, the $i$-th encryptor and the decryptor, respectively. 
    
    \item \textbf{NDD-FE.KeyDerive($pk_1, sk_1, \{pk_{2i}\}_{i=1,2,\dots,I}, ctr, \boldsymbol{y},$\\$ aux$)$\rightarrow$$sk_\otimes$}: It is executed by the generator. It takes $(pk_1, sk_1)$, the $\{pk_{2i}\}_{i=1,2,\dots,I}$ of $I$ encryptors, an incremental counter $ctr$, a vector $\boldsymbol{y}$ and auxiliary information $aux$ as input, and outputs the functional key $sk_\otimes$.
    
    \item \textbf{NDD-FE.Encrypt($pk_1, sk_{2i}, pk_3, ctr, x_i, aux$)$\rightarrow$$c_i$}: It is executed by the encryptor. It takes $pk_1, (pk_{2i}, sk_{2i}),$ $pk_3, ctr, aux$, and the data $x_i$ as input, and outputs the ciphertext $c_i = pk_1^{r_i^{ctr}} \cdot pk_3^{x_i}$, where $r_i^{ctr}$ is generated by $H_1$.
    
    \item \textbf{NDD-FE.Decrypt($pk_1, sk_\otimes, sk_3, \{ct_i\}_{i=1,2,\dots,I}, \boldsymbol{y}$)$\rightarrow$\\$\langle \boldsymbol{x}, \boldsymbol{y} \rangle$}: It is executed by the decryptor. It takes $pk_1, sk_\otimes,$ $sk_3$, $\{ct_i\}_{i=1,2,\dots,I}$ and $\boldsymbol{y}$ as input. First, it outputs $g^{\langle \boldsymbol{x}, \boldsymbol{y} \rangle}$ and subsequently calculates $\log(g^{\langle \boldsymbol{x}, \boldsymbol{y} \rangle})$ to reconstruct the result of the inner product of $( \boldsymbol{x}, \boldsymbol{y} )$.
 
\end{itemize}

\subsection{The Proposed Decentralized Efficient Functional Encryption Scheme}   
We propose a decentralized efficient functional encryption (DEFE) scheme for more secure and efficient inner product operations. Our DEFE is an adaptation of NDD-FE in three aspects: 
\begin{itemize}
    \item \textbf{Decentralized authority:} DEFE eliminates reliance on the third-party entities (e.g., the generator) by enabling encryptors to jointly generate the decryptor's decryption key. 
    \item  \textbf{Mix-and-Match attack resistance:} DEFE inherently restricts the decryptor from obtaining the true inner product results, which prevents the decryptor from launching inference attacks.
    \item  \textbf{Efficient decryption:} DEFE enables efficient decryption by modifying the ciphertext structure. This avoids the costly discrete logarithm computations in NDD-FE.
\end{itemize}
We consider our SRFed system with one decryptor (i.e., the server) and $I$ encryptors (i.e., the clients). The $i$-th encryptor encrypts the $i$-th component $x_i$ of the $I$-dimensional message vector $\boldsymbol{x}$. The message vector $\boldsymbol{x}$ and key vector $\boldsymbol{y}$ satisfy $\|\boldsymbol{x}\|_\infty \le X$ and $\|\boldsymbol{y}\|_\infty \le Y$, with $X \cdot Y < N$, where $N$ is the Paillier composite modulus \cite{paillier}. Decryption yields $\langle \boldsymbol{x}, \boldsymbol{y} \rangle \bmod N$, which equals the integer inner product $\langle \boldsymbol{x}, \boldsymbol{y} \rangle$ under these bounds. Let $M = \left\lfloor \frac{1}{2} \left( \sqrt{\frac{N}{I}} \right)\right\rfloor$. We assume $X, Y < M$ in DEFE. Specifically, the construction of the DEFE scheme is as follows. The notations are described in Table ~\ref{tab:notations}. 

\begin{table}
	\caption{Notation Descriptions}
    \resizebox{\columnwidth}{!}{
	\label{tab:notations}
		\begin{tabular}{cccc}
			\toprule
			Notations & Descriptions&Notations & Descriptions\\
			\midrule
			$pk, sk$ & Public/secret key & $skf$ & Functional key    \\
			$T$ & Total training round &  $t$ & Training round \\
            $I$ & Number of clients & $C_i$ & The $i$-th client \\
            $D_i$ & Dataset of $C_i$ &  $D_i^*$ & Poisoned dataset \\
            $l$ &  Model layer      &    $\zeta$ & Length of $W_t$ \\
			$W_t$ & Global model &    $W_{t+1}$ & Aggregated model\\
            $W_t^i$ & Benign model & $(W_{t}^{i})^*$ & Poisonous model\\
            $|W_t^{(l)}|$ & Length of $W_{t}^{l}$  & $\lVert W_t^{(l)} \rVert$ & The Euclidean norm of $\lVert W_t^{(l)} \rVert$\\
			$\eta$ & Hash noise &  $H_1$ & Hash function  \\
			  $noise$ &  Gaussian noise    & $E_{t}^{i}$ & Encrypted update\\
            $V_t^i$ & projection vector & $OA$  & Overall accuracy \\
			$SA$ & Source accuracy & $ASR$ & Attack success rate \\
			\bottomrule
	\end{tabular}
    }
\end{table}

\begin{itemize}
	\item
	\(
	\textbf{DEFE.Setup}(1^\lambda, X, Y) \rightarrow pp
	\):
	It takes the security parameter $1^\lambda$ as input and outputs the public parameters $pp$, which include the modulus $N$, generator $g$, and hash function $H_1$. 
	It initializes by selecting safe primes $p = 2p' + 1$ and $q = 2q' + 1$ with $p', q' > 2^{l(\lambda)}$ (where $l$ is a polynomial in $\lambda$), ensuring the factorization hardness of $N = pq$ is $2^\lambda$-hard and $N > XY$. A generator $g'$ is uniformly sampled from $\mathbb{Z}_{N^2}^*$, and $g = g'^{2N} \mod N^2$ is computed to generate the subgroup of $(2N)$-th residues in $\mathbb{Z}_{N^2}^*$. Hash function $H_1: \mathbb{Z} \times \mathbb{N} \times \mathcal{AUX} \rightarrow \mathbb{Z}$ is defined, where $\mathcal{AUX}$ denotes auxiliary information (e.g., task identifier, timestamp).

\item
	\(
	\textbf{DEFE.KeyGen}(1^\lambda, N,g) \rightarrow (pk_i,   sk_{i})
	\):
        It is executed by $n$ encryptors. It takes $\lambda$, $N$, and $g$ as input, and outputs the corresponding key pair $(pk_i, sk_i)$. For the $i$-th encryptor, an integer $s_i$ is drawn from a discrete Gaussian distribution $D_{\mathbb{Z}, \sigma}$ ($\sigma > \sqrt{\lambda} \cdot N^{5/2}$), and the public key $h_i = g^{s_i} \mod N^2$, forming key pair $(pk_i = h_i, sk_i = s_i)$.

\item
	\(
	\textbf{DEFE.Encrypt}(pk_i, sk_i, ctr, x_i, aux) \rightarrow ct_i
	\): It is executed by $I$ encryptors. It takes key pair $(pk_i, sk_i)$, counter $ctr$, data $x_i \in \mathbb{Z}$, and $aux$ as input, and outputs noise-augmented ciphertext $ct_i \in \mathbb{Z}_{N^2}$. Considering the multi-round training process of FL, each $i$-th encryptor generates a noise value $\eta_{t,i}$ for the $t$-th round following the recursive relation $\eta_{t,i} = H_1(\eta_{t-1,i}, pk_i, ctr) \mod M$, where $ctr$ is an incremental counter. The initial noise $\eta_{0,i}$ is uniformly set across all encryptors via a single communication. Using the noise-augmented data $x'_i = x_i + \eta_{t,i}$ and secret key $sk_i$, the encryptor computes the ciphertext $ct_i^{\prime} = (1+N)^{x'_i} \cdot g^{r_i^{ctr}} \mod N^2$ with $r_i^{ctr} = H_1(sk_i, ctr, aux)$ and $aux \in \mathcal{AUX}$.   

\item
\(
\textbf{DEFE.FunKeyGen}\bigl( (pk_i, sk_i)_{i=1}^I, ctr, y_i, aux \bigr) \rightarrow skf_{i,\boldsymbol{y}}
\): It is executed by $I$ encryptors. Each encryptor computes its partial functional key \(skf_{i,\boldsymbol{y}}\). It takes the public/secret key pairs ${(pk_i, sk_i)}^{n}_{i=1}$ of encryptors and the 
	$i$-th component $y_i$ of the key vector $\boldsymbol{y}$ as inputs and outputs:
	\begin{equation}
	    skf_{i,\boldsymbol{y}} = r_i^{ctr} y_i + \sum\nolimits_{j=1}^{i-1} \varphi^{i,j} - \sum\nolimits_{j=i+1}^{I} \varphi^{i,j},
	\end{equation}
	where $r_i^{ctr} = H_1(sk_i, ctr, aux)$ and $\varphi^{i,j} = H_1(pk_j^{sk_i}, ctr,$ $aux) $. Note that \( \varphi^{i,j} = \varphi^{j,i} \).


\item
\(
\textbf{DEFE.FunKeyAgg}\bigl(\{skf_{i,\boldsymbol{y}}\}^{I}_{i=1}) \rightarrow skf_{\boldsymbol{y}}
\): It is executed by the decryptor. It inputs partial functional keys \( skf_{i,\boldsymbol{y}} \) and derives the final functional key:
\begin{equation}
skf_{\boldsymbol{y}} = \sum\nolimits_{i=1}^{I}skf_{i,\boldsymbol{y}} = \sum\nolimits_{i=1}^{I}r_i^{ctr}\cdot y_i \in \mathbb{Z}.
\end{equation}

\item
	\(
	\textbf{DEFE.AggDec}(skf_{\boldsymbol{y}}, \{ct_i\}^{I}_{i=1}) \rightarrow \langle \boldsymbol{x'}, \boldsymbol{y} \rangle
	\): It is executed by the decryptor. It first computes
	\begin{equation}
    CT_{\boldsymbol{x'}}= \left( \prod\nolimits_{i=1}^{I} ct_i^{y_i} \right) \cdot g^{-skf_{\boldsymbol{y}}} \mod N^2.
    \end{equation}
	Then, it outputs $\log_{(1+N)}(CT_{\boldsymbol{x'}}) = \frac{CT_{\boldsymbol{x'}} - 1 \mod N^2}{N} = \langle \boldsymbol{x'}, \boldsymbol{y} \rangle.$
	
\item\(	\textbf{DEFE.UsrDec}(\langle \boldsymbol{x'}, \boldsymbol{y} \rangle,\{pk_{i}\}^I_{i=1}, \boldsymbol{y}, ctr) \rightarrow \langle \boldsymbol{x}, \boldsymbol{y} \rangle	\):
	It is executed by $I$ encryptors. During FL processes, each encryptor maintains a \( I \)-dimensional noise list \( \texttt{$L_t$} = [\eta_{t,i}]^I_{i=1} \) for each training round $t$. Based on this, each encryptor can obtain the true inner product value: $\langle \boldsymbol{x}, \boldsymbol{y} \rangle = \langle \boldsymbol{x'}, \boldsymbol{y} \rangle - \sum_{i=1}^{I}\eta_{t,i}\cdot y_i.$
    
	
\end{itemize}

\section{System Design}

\begin{figure*}[t]
\centering
\includegraphics [width=0.8\textwidth]{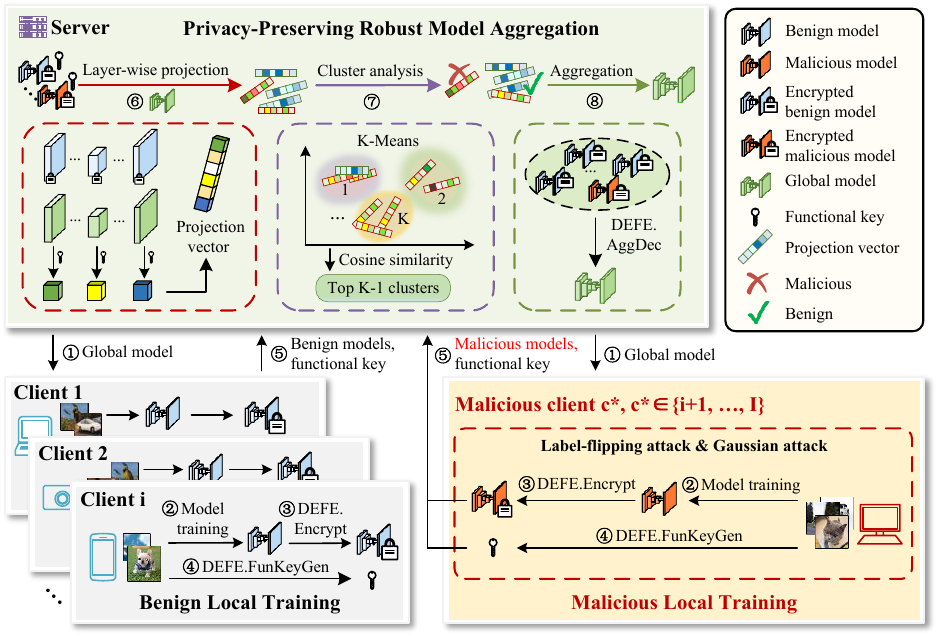}
\caption{The workflow of SRFed.}
\label{SRFed_workflow}
\end{figure*}

\subsection{High-Level Description of SRFed}
The workflow of SRFed is illustrated in Figure \ref{SRFed_workflow}. Specifically, SRFed iteratively performs the following three
steps: 1) \textit{Initialization}. The server initializes the global model $W_0$ and distributes it to all clients (step \ding{172}). 2) \textit{Local training}. In the $t$-th training iteration, each client $C_i$ receives the global model $W_t$ from the server and performs local training on its private dataset to obtain the local model $W_t^i$ (step \ding{173}). To protect model privacy, $C_i$ encrypts the local model and gets the encrypted model $E_t^i$ (step \ding{174}). Then, $C_i$ generates the functional key $skf_t^i$ for model detection (step \ding{175}), and uploads the encrypted model and the functional key to the server (step \ding{176}). 3) \textit{Privacy-preserving robust model aggregation}. Upon receiving all encrypted local models $\{E_t^i\}_{i=1,...,I}$, the server computes a layer-wise projection vector $V _t^i$ for each model based on the global model $W_t$ (step \ding{177}). The server then performs clustering analysis on $\{V_t^i\}_{i=1,...,I}$ to filter malicious models and identify benign models (step \ding{178}). Finally, the server aggregates these benign clients to update the global model $W_{t+1}$ (step \ding{179}).


\subsection{Construction of SRFed}
\subsubsection{Initialization}
    In this phase, the server first executes the 
	$\textbf{DEFE.Setup}(1^\lambda, X,$ $Y)$ algorithm  to generate the public parameters $pp$, which are then made publicly available. Each client \( C_i \) (\( i \in [1,I] \)) subsequently generates its key pair $(pk_i, sk_i)$ by executing the \(
	\textbf{DEFE.KeyGen}(1^\lambda, N, g) 
    \) algorithm. Finally, the server distributes the initial global model $W_0$ to all clients.
\subsubsection{Local Training}
The local training phase consists of three components: model training, model encryption, and functional key generation.
\textit{1) Model Training:} 
In the $t$-th ($t \in [1,T]$) training round, once receiving the global model \( W_t \), each client \( C_i \) utilizes its local dataset \( D_i \) to update \( W_t \) and obtains model update \( W_t^i \). For benign clients, they minimize their local objective function \( L_i \) to obtain \( W_t^i \), i.e.,
\begin{equation}
W_t^i = \arg\min_{W_t} L_i(W_t, D_i).
\label{fine-tune}
\end{equation}
Malicious clients execute distinct update strategies based on their attack method. Specifically, to perform a Gaussian attack, malicious clients first conduct normal local training according to Equation~\eqref{fine-tune} to obtain the benign model \( W_t^i \), and subsequently inject Gaussian noise \( noise_t \) into it to produce a poisoned model $(W_t^i)^*$, i.e.,
\begin{equation}
(W_t^i)^* = W_t^i + noise_t.
\end{equation}
In addition, to launch a label-flipping attack, malicious clients first poison their training datasets by flipping all samples labeled as \( l_{\text{src}} \) to a target class \( l_{\text{tar}} \). They then perform local training on the poisoned dataset \( D_i^*(l_{\text{src}} \rightarrow l_{\text{tar}}) \) to derive the poisoned model $(W_t^i)^*$, i.e.,
\begin{equation}
(W_t^i)^* = \arg\min_{W_t^i} L_i\left(W_t^i, D_i^*(l_{\text{src}} \rightarrow l_{\text{tar}})\right).
\label{eq:malicious_model_update}
\end{equation}

\textit{2) Model Encryption:}
To protect the privacy information of the local model, the clients exploit the DEFE scheme to encrypt their local models. Specifically, the clients first parse the local model as \( W_t^i = [W_t^{(i,1)}, \dots, W_t^{(i,l)}, \dots, W_t^{(i,L)}] \), 
where \( W_t^{(i,l)} \) denotes the parameter set of the \( l \)-th model layer. For each parameter element \( W_t^{(i,l)}[\varepsilon] \) in \( W_t^{(i,l)} \), client \( C_i \) executes the $\textbf{DEFE.Encrypt}(pk_i, sk_i, ctr, W_t^{(i,l)}[\varepsilon])$ algorithm to generate the encrypted parameter \( E_t^{(i,l)}[\varepsilon] \), i.e.,
\begin{equation}\label{encrypt}
    E_t^{(i,l)}[\varepsilon] = (1+N)^{{W_t^{(i,l)}}'[\varepsilon]} \cdot g^{{r_i}^{ctr}} \bmod N^2,
\end{equation}
where ${W_t^{(i,l)}}'[\varepsilon]$ is the parameter \( W_t^{(i,l)}[\varepsilon] \) perturbed by a noise term \( \eta \), i.e., 
\begin{equation}
    {W_t^{(i,l)}}'[\varepsilon] = W_t^{(i,l)}[\varepsilon] + \eta.
\end{equation}
In SRFed, the noise \( \eta \) remains fixed for all clients during the first \( T-1 \) training rounds, and is set to $0$ in the final training round. Specifically, \( \eta \) is an integer generated by client \( C_1 \) using the hash function \( H_1 \) during the first iteration, and is subsequently broadcast to all other clients for model encryption. The magnitude of \( \eta \) is constrained by:
\begin{equation}\label{eta}
    m_l \cdot \eta^2 \ll \lVert W_0^{(l)} \rVert^2, 
    \quad \forall l \in [1, L],
\end{equation}
where \( W_0^{(l)} \) denotes the \( l \)-th layer model parameters of the initial global model \( W_0 \), and \( m_l \) represents the number of parameters in \( W_0^{(l)} \). Finally, the client $C_i$ obtains the encrypted local model $E_t^i = [E_t^{(i,1)}, \dots, E_t^{(i,l)}, \dots, E_t^{(i,L)}]$.



\textit{3) Functional Key Generation:} Each client \( C_i \) generates a functional key vector $skf_t^i = [skf_t^{(i,1)}, skf_t^{(i,2)}, \dots, skf_t^{(i,L)}]$ to enable the server to perform model detection on encrypted models. Specifically, for the $l$-th layer of the global model \( W_t \), client \( C_i \) executes the $\textbf{DEFE.FunKeyGen}(pk_i, sk_i, ctr, W_t^{(l)}[\varepsilon])$ algorithm to generate element-wise functional keys, i.e.,
\begin{equation}
skf_t^{(i,l,\varepsilon)} = r_i^{ctr} \cdot W_t^{(l)}[\varepsilon] = H_1(sk_i, ctr, \varepsilon) \cdot W_t^{(l)}[\varepsilon],
\end{equation}
where \( W_t^{(l)}[\varepsilon] \) denotes the \(\varepsilon\)-th parameter of \( W_t^{(l)} \). After processing all elements in \( W_t^{(l)} \), client \( C_i \) obtains the set of element-wise functional keys $\{\{skf_t^{(i,l,\varepsilon)} \}_{\varepsilon=1}^{|W_t^{(l)}|}\}_{l=1}^{L}$. Subsequently, the layer-level functional key is derived by aggregating the element-wise keys using the 
\(
\textbf{DEFE.FunKeyAgg}(\{ skf_t^{(i,l,\varepsilon)} \}_{\varepsilon=1}^{|W_t^{(l)}|})
\) algorithm, i.e.,
\begin{equation}
skf_t^{(i,l)} = \sum\nolimits_{\varepsilon=1}^{|W_t^{(l)}|} skf_t^{(i,l,\varepsilon)}.
\end{equation}
This procedure is repeated for all layers to obtain the complete functional key vector \( skf_t^i \) for client \( C_i \). Finally, each client uploads the encrypted local model \( E_t^i \) and the corresponding functional key \( skf_t^i \) to the server for subsequent model detection and aggregation.

\subsubsection{Privacy-Preserving Robust Model Aggregation} 

To resist poisoning attacks from malicious clients, SRFed implements a privacy-preserving robust aggregation strategy, which enables secure detection and aggregation of encrypted local models without exposing private information. As illustrated in Figure \ref{SRFed_workflow}, the proposed method performs layer-wise projection and clustering analysis to identify abnormal updates and ensure reliable model aggregation. Specifically, in each training round, the local model $W_t^i$ and the global model $W_t$ are decomposed layer by layer. For each layer, the parameters are projected onto the corresponding layer of the global model, and clustering is performed on the projection vectors to detect anomalous models. After that, the server filters malicious models and aggregates the remaining benign updates. Unlike prior defenses that rely on global statistical similarity between model updates \cite{foolsgold, HE_shieldfl, Median_original}, our approach captures fine-grained parameter anomalies and conducts clustering analysis to achieve effective detection even under non-IID data distributions.


\textit{1) Model Detection}: Once receiving $E_t^i$ and $skf_t^i$ from client $C_i$, the server computes the projection $V_t^{(i,l)}$ of $W_t^{(i,l)'}$ onto $W_t^{(l)}$, i.e.,
 \begin{equation}
     V_t^{(i,l)} = \frac{\langle W_t^{(i,l)'}, W_t^{(l)} \rangle}{\lVert W_t^{(l)} \rVert_2}.
 \end{equation}
 Specifically, the server first executes the $\textbf{DEFE.AggDec}(skf_t^{(i,l)},$ $E_t^{(i,l)})$ algorithm, which effectively computes the inner product of $W_t^{(i,l)'}$ and $W_t^{(l)}$. This value is then normalized by $\lVert W_t^{(l)} \rVert_2$ to obtain
\begin{equation}
    V_t^{(i,l)} = \frac{\textbf{DEFE.AggDec}(skf_t^{(i,l)}, E_t^{(i,l)})}{\lVert W_t^{(l)} \rVert_2}.
\end{equation}
By iterating over all $L$ layers, the server obtains the layer-wise projection vector $V_t^i = [V_t^{(i,1)}, V_t^{(i,2)}, \dots, V_t^{(i,L)}]$ corresponding to client $C_i$. After computing projection vectors for all clients, the server clusters the set $\{V_t^i\}_{i=1}^I$ into $K$ clusters $\{\Omega_1, \Omega_2, \dots, \Omega_K\}$ using the K-Means algorithm. For each cluster $\Omega_k$, the centroid vector $\bar{V}_k$ is computed, and the average cosine similarity $\overline{cs}_k$ between all vectors in the cluster and $\bar{V}_k$ is calculated. Finally, the $K-1$ clusters with the largest average cosine similarities are identified as benign clusters, while the remaining cluster is considered potentially malicious.

\textit{3) Model Aggregation}: The server first maps the vectors in the selected $K-1$ clusters to their corresponding clients, generating a client list $L^t_{bc}$ and a weight vector $\gamma_t = (\gamma^1_t, \dots, \gamma^I_t)$, where
\begin{equation}
    \gamma^i_t = \begin{cases} 
1 & \text{if } C_i \in L^t_{bc}, \\ 
0 & \text{otherwise.}
\end{cases}
\end{equation}
The server then distributes $\gamma_t$ to all clients. Upon receiving $\gamma^i_t$, each client $C_i$ locally executes the $\textbf{DEFE.FunKeyGen}(pk_i, sk_i,$ $ctr, \gamma^i_t, aux)$ algorithm to compute the partial functional key $skf^{(i,\mathsf{Agg})}_t$ as
\begin{equation}
    skf^{(i,\mathsf{Agg})}_t = r_i^{ctr} y^i_t + \sum_{j=1}^{i-1} \varphi^{i,j} - \sum_{j=i+1}^{n} \varphi^{i,j}.
\end{equation}
Each client uploads $skf^{(i,\mathsf{Agg})}_t$ to the server. Subsequently, the server executes the $\textbf{DEFE.FunKeyAgg}\left( (skf^{(i,\mathsf{Agg})}_t)_{i=1}^I \right)$ to compute the aggregation key as 
\begin{equation}
    skf^{\mathsf{Agg}}_t = \sum_{i=1}^{I} skf^{(i,\mathsf{Agg})}_t.
\end{equation}
Finally, the server performs layer-wise aggregation to obtain the noise-perturbed global model $W_{t+1}'$ as
\begin{align}
    W_{t+1}^{(l)'}[\varepsilon]  
        &= \frac{\text{DEFE.AggDec}\left( skf^{\mathsf{Agg}}_t, \{E_t^{(i,l)}[\varepsilon]\}_{i=1}^I \right)}{n} \\ 
        &= \frac{\langle (W_{t}^{(1,l)'}[\varepsilon], \dots, W_{t}^{(I,l)'}[\varepsilon]), \gamma_t \rangle}{n} \notag
\end{align}
where $n$ denotes the number of 1-valued elements in $L^t_{bc}$. The server then distributes $W_{t+1}'$ to all clients for the $(t+1)$-th training round. Note that $W_{t+1}'$ is noise-perturbed, the clients must remove the perturbation to recover the accurate global model $W_{t+1}$. They will execute the $\textbf{DEFE.UsrDec}(W_{t+1}^{(l)'}[\varepsilon], \gamma_t)$ algorithm to restore the true global model parameter $W_{t+1}^{(l)}[\varepsilon] = W_{t+1}^{(l)'}[\varepsilon] - \eta$.

\section{Analysis}\label{sec:analysis}

\subsection{Confidentiality}

In this subsection, we demonstrate that our DEFE-based SRFed framework guarantees the confidentiality of clients’ local models under the Honest-but-Curious (HBCS) security setting. 
\begin{definition}[Decisional Composite Residuosity (DCR) Assumption \cite{DCR}]\label{DCR}
    Selecting safe primes $p = 2p'+1$ and $q = 2q'+1$ with $p',q'>2^{l(\lambda)}$, where $l$ is a polynomial in security parameter $\lambda$, let $N=pq$. The Decision Composite Residuosity (DCR) assumption states that, for any Probability Polynomial Time (PPT) adversary $\mathcal{A}$  and any distinct inputs $x_0, x_1$, the following holds: $$|Pr_{win}(\mathcal{A}, (1+N)^{x_0}\cdot g^{r_i^{ctr} }\mod N^2,x_0,x_1) - \frac{1}{2}| = negl(\lambda),$$
where $Pr_{win}$ denotes the probability that the adversary $\mathcal{A}$ distinguishes ciphertexts.
\end{definition}
\begin{definition}[Honest but Curious Security (HBCS)]\label{HBCS} Consider the following game between an adversary $\mathcal{A}$ and a PPT simulator $\mathcal{A}^*$, a protocol $\Pi$ is secure if the real-world view $\textbf{REAL}_{\mathcal{A}}^{\Pi}$ of $\mathcal{A}$ is computationally indistinguishable from the ideal-world view $\textbf{IDEAL}_{\mathcal{A}^*}^{\mathcal{F_{\Pi}}}$ of $\mathcal{A}^*$, i.e., for all inputs $\hat{x}$ and intermediate results $\hat{y}$ from participants, it holds $
\textbf{REAL}_{\mathcal{A}}^{\Pi}(\lambda, \hat{x}, \hat{y})  \overset{c}{\equiv} \textbf{IDEAL}_{\mathcal{A}^*}^{\mathcal{F_\Pi}}(\lambda, \hat{x}, \hat{y})$,
where $\overset{c}{\equiv}$ denotes computationally indistinguishable.
\end{definition}



\begin{theorem}
   SRFed achieves Honest but Curious Security under the DCR assumption, which means that for all inputs $\{C_t^i,  {skf}_{t}^{i} \}_{i=1,...,I}$ and intermediate results ($V_t^i$, $W_{t+1}'$, $W_T$), SRFed holds: 
   
   $\textbf{REAL}_{\mathcal{A}}^{SRFed}(C_t^i,  {skf}_{t}^{i},skf_t^\mathsf{Agg},V_t^i,W_{t+1}',W_T) \overset{c}{\equiv}$ 
   
   $\textbf{IDEAL}_{\mathcal{A}^*}^{\mathcal{F}_{SRFed}}(C_t^i,  {skf}_{t}^{i},skf_t^\mathsf{Agg},V_t^i,W_{t+1}',W_T)$. 
\end{theorem}

\begin{proof}
    
To prove the security of SRFed, we just need to prove the confidentiality of the privacy-preserving defense strategy, since only it involves the computation of private data by unauthorized entities (i.e., the server). For the curious server, $\textbf{REAL}_{\mathcal{A}}^{SRFed}$ contains intermediate parameters and encrypted local models $\{C_t^i\}_{i=1,...,I}$ collected from each client during the execution of SRFed protocols. Besides, we construct a PPT simulator $\mathcal{A}^*$ to execute $\mathcal{F}_{SRFed}$, which simulates each process of the privacy-preserving defensive aggregation strategy. The detailed proof is described below.

\textbf{Hyb}$_1$ We initialize a series of random variables whose distributions are indistinguishable from $\textbf{REAL}_{\mathcal{A}}^{SRFed}$ during the real protocol execution.

\textbf{Hyb}$_2$  In this hybrid, we change the behavior of simulated client $C_i$ $(i \in [1,I])$. $C_i$ takes the selected random vector of random variables $\Theta_{W}$ as the local model $W_t^{i'}$, and uses the DEFE.Encrypt algorithm to encrypt $W_t^{i'}$. As only the original contents of ciphertexts have changed, it guarantees that the server cannot distinguish the view of $\Theta_{W}$ from the view of original $W_t^{i'}$ according to the Definition \eqref{DCR}. Then, $C_i$ uses the DEFE.FunKeyGen algorithm to generate the key vector $skf_t^i = [skf_t^{(i,1)},skf_t^{(i,2)},\dots,skf_t^{(i,L)}].$ Note that each component of $skf_t^i$ essentially is the inner product result, thus revealing no secret information to the server.

\textbf{Hyb}$_3$  In this hybrid, we change the input of the protocol of \textit{Secure Model Aggregation} executed by the server with encrypted random variables instead of real encrypted model parameters. The server gets the plaintexts vector $V_t^i = [V_t^{(i,1)},V_t^{(i,2)},\dots,V_t^{(i,L)}]$ corresponding to $C_i$, which is the layer-wise projection of $\Theta_{W}$ and $W_{t}$. As the inputs $\Theta_{W}$ follow the same distribution as the real $W^{i'}_t$, the server cannot distinguish the $V_t^i$ between ideal world and real world without knowing further information about the inputs. Then, the server performs clustering based on $\{V_t^i\}^I_{i=1}$ to obtain $\{\Omega_k\}^K_{k=1}$. Subsequently, it computes the average cosine similarity of all vectors within each cluster to their centroid, and assigns client weights accordingly. Since $\{V_t^i\}^I_{i=1}$ is indistinguishable between the ideal world and the real world, the intermediate variables calculated via $\{V_t^i\}^I_{i=1}$ above also inherit this indistinguishability. Hence, this hybrid is indistinguishable from the previous one.

\textbf{Hyb}$_4$ In this hybrid, the aggregated model $W_{t+1}'$ is computed by the DEFE.AggDec algorithm. $\mathcal{A}^*$ holds the view $\textbf{IDEAL}_{\mathcal{A}^*}^{\mathcal{F}_{SRFed}}$$=$$(C_t^i,  {skf}_{t}^{i},skf_t^\mathsf{Agg},V_t^i,W_{t+1}',W_T)$, where $skf_t^{\mathsf{Agg}}$ is obtained by the interaction of non-colluding clients and server, the full security property of DEFE and the non-colluding setting ensure the security of $skf_t^{\mathsf{Agg}}$. Among the elements of intermediate computation, the local model $W_t^{i'}$ is encrypted, which is consistent with the previous hybrid. Throughout the $T$-round iterative process, the server obtains the noise-perturbed aggregated model $W_{t+1}' = W_{t+1} + \eta $ via the secure model aggregation when $0\le t <T-1$. Thus, the server cannot infer the real $W_{t+1}$, and cannot distinguish the $W_{t+1}'$ between ideal world and real world. When $t=T-1$, since the distribution of $\Theta_{W}$ remains identical to that of $W_T$, the probability that the server can distinguish the final averaged aggregated model $W_{T}$ is negligible. Hence, this hybrid is indistinguishable from the previous one.

\textbf{Hyb}$_5$ When $0 \le t<T$, all clients further execute the DEFE.UsrDec algorithm to restore the $W_{t+1}$. This process is independent of the server, hence this hybrid is indistinguishable from the previous one.

The argument above proves that the output of $\textbf{IDEAL}_{\mathcal{A}^*}^{\mathcal{F}_{SRFed}}$ is indistinguishable from the output of $\textbf{REAL}_{\mathcal{A}}^{SRFed}$. Thus, it proves that SRFed guarantees HBCS.

\end{proof} 

\subsection{Robustness}


 
To theoretically analyze the robustness of SRFed against poisoning attacks, we first prove the following theorem.
\begin{theorem}
When the noise perturbation \( \eta \) satisfies the constraint in~\eqref{eta}, the clustering results of SRFed over all \( T \) iterations remain approximately equivalent to those obtained using the original local models $\{W_t^i\}_{1, 2, ..., I}$.

\end{theorem}
 \begin{proof}
Let $\overline{\eta}$ be a vector of the same shape as $W_t^{(i,l)'}$ with all entries equal to $\eta$, and ${V_t^i}_{real}$ be the real projection vector derived from the noise-free models. We discuss the following three cases.

\ding{172}
 $t=0:$ For any $i\in[1,I]$, $W_0^{(i,l)'} = W_0^{(i,l)} + \overline{\eta}$, we have
\begin{equation}
\begin{aligned}[t]  
V_0^i &= \frac{\langle W_0^{(i,l)'}, W_0^{(l)}\rangle}{\lVert W_0^{(l)} \rVert_2} = \frac{\langle W_0^{(i,l)} + \overline{\eta}, W_0^{(l)}\rangle}{\lVert W_0^{(l)} \rVert_2} = {V_0^i}_{real} +  \frac{\langle \overline{\eta}, W_0^{(l)}\rangle}{\lVert W_0^{(l)} \rVert_2}.\\  
\end{aligned}
\end{equation}
Note that $\frac{\langle \overline{\eta}, W_0^{(l)}\rangle}{\lVert W_0^{(l)} \rVert_2}$ is identical for any client, the clustering result of $\{V_0^i\}^I_{i=1}$ is entirely equivalent to that of $\{{V_0^i}_{real}\}^I_{i=1}$ based on the underlying computation of K-Means.

\ding{173} 
 $0<t<T:$ For any $i\in[1,I]$, $W_t^{(i,l)'} = W_t^{(i,l)} + \overline{\eta}$. Correspondingly, $W_t^{(l)'} = W_t^{(l)} + \overline{\eta}$, and we have 
\begin{equation}
\begin{split}
V_t^i &= \frac{\langle W_t^{(i,l)'}, W_t^{(l)'}\rangle}{\lVert W_t^{(l)'} \rVert_2} \\
      &= \frac{\langle W_t^{(i,l)} , W_t^{(l)}\rangle+\langle W_t^{(i,l)} , \overline{\eta}\rangle+\langle  \overline{\eta}, W_t^{(l)}\rangle+\langle \overline{\eta},\overline{\eta}\rangle}{\sqrt[]{ \lVert W_t^{(l)} \rVert_2^2+2\langle W_t^{(l)},\overline{\eta}\rangle+\lVert\overline{\eta}\rVert_2^2}}.
\end{split}
\end{equation}
By combining the above equation with the constraint \eqref{eta}, $V_t^i$ is approximately equivalent to the real value of ${V_t^i}_{real}$.

\ding{174} 
$t=T:$ For any $i\in[1,I]$, $W_T^{(i,l)'} = W_T^{(i,l)}$. Correspondingly, $W_T^{(l)'} = W_T^{(l)} + \overline{\eta}$, and we have 
 \begin{equation}
\begin{aligned}[t]  
V_T^i &= \frac{\langle W_T^{(i,l)}, W_T^{(l)'}\rangle}{\lVert W_T^{(l)'} \rVert_2} = \frac{\langle W_T^{(i,l)} , W_T^{(l)}\rangle+\langle W_T^{(i,l)} , \overline{\eta}\rangle}{\sqrt[]{ \lVert W_T^{(l)} \rVert_2^2+2\langle W_T^{(l)},\overline{\eta}\rangle+\lVert\overline{\eta}\rVert_2^2}}. \\  
\end{aligned}
\end{equation}
Similarly, by combining the above equation with the constraint \eqref{eta}, $V_T^i$ is approximately equivalent to the real value of ${V_T^i}_{real}$.

Therefore, across all iterations, the clustering results based on \( \{V_t^i\} \) closely approximate those derived from the original local models, confirming that the introduced perturbation does not affect model detection.
 \end{proof}



Then, we introduce a key assumption, which has been proved in \cite{pbfl,NIFE_2024_BSR_FL}. This assumption reveals the essential difference between malicious and benign models and serves as a core basis for subsequent robustness analysis.
\begin{assumption}\label{tau}
    An error term $\tau^{(t)}$ exists between the average malicious gradients $\mathbf{W}_t^{i*}$ and the average benign gradients $\mathbf{W}_t^i$ due to divergent training objectives. This is formally expressed as:
    \begin{equation}
        \sum_{C_i \in \mathcal{M}} \mathbf{W}_t^{i*} = \sum_{C_i \in \mathcal{B}} \mathbf{W}_t^i + \tau^{(t)}.
    \end{equation}
    The magnitude of $\tau^{(t)}$ exhibits a positive correlation with the number of iterative training rounds.
\end{assumption}

\begin{theorem}
    SRFed guarantees robustness to malicious clients in non-IID settings, provided that most clients are benign.
\end{theorem}
\begin{proof}

In the secure model aggregation phase of SRFed, the server collects the encrypted model $C_t^i$ and the corresponding key vectors $skf_t^i$ from each client, then computes the projection $V_t^{(i,l)} $ of  $W_t^{(i,l)}$ onto $W_t^{(l)}$, i.e., $\frac{\langle W_t^{(i,l)}, W_t^{(l)}\rangle}{\lVert W_t^{(l)} \rVert_2}$. By iterating over $L$ layers, the server obtains the layer-wise projection vector $V_t^i = [V_t^{(i,1)},V_t^{(i,2)},\dots,V_t^{(i,L)}]$ corresponding to $C_i$. Subsequently, the server performs clustering on the projection vectors $\{V_t^i\}^I_{i=1}$. Based on the Assumption \eqref{tau}, a non-negligible divergence $\tau^{(t)}$ emerges between benign and malicious local models, which grows with the number of iterations. Meanwhile, by independently projecting each layer's parameters onto the corresponding layer of the global model, our operation eliminates cross-layer interference. This ensures that malicious modifications confined to specific layers can be detected significantly more effectively. Therefore, our clustering approach successfully distinguishes between benign and malicious models by grouping them into separate clusters. Due to significant distribution divergence, malicious models exhibit a lower average cosine similarity to their cluster center. Consequently, our scheme filters out the cluster containing malicious models by computing average cosine similarity, ultimately achieving robust to malicious clients.
\end{proof}


\subsection{Efficiency}
\begin{theorem}
The computation and communication complexities of SRFed are $\mathcal{O}(T_{lt})+\mathcal{O}(\zeta T_{me-defe})+\mathcal{O}(T_{md-defe})+ \mathcal{O}(T_{ma-defe})$ and $\mathcal{O}(I\zeta|w_{defe}|)+\mathcal{O}(IL|w|)$, respectively.
\end{theorem}
\begin{proof}

To evaluate the efficiency of SRFed, we analyze its computational and communication overhead per training iteration and compare it with ShieldFL \cite{HE_shieldfl}. ShieldFL is an efficient PPFL framework based on the partially homomorphic encryption (PHE) scheme. The comparative results are presented in Table \ref{tab:com_comm}.
Specifically, the computational overhead of SRFed comprises four components: local training $\mathcal{O}(T_{lt})$, model encryption $\mathcal{O}(\zeta T_{me-defe})$, model detection $\mathcal{O}(T_{md-defe})$, and model aggregation $\mathcal{O}(T_{ma-defe})$. For model encryption, FE inherently offers a lightweight advantage over PHE, leading to $\mathcal{O}(\zeta T_{me-defe}) < \mathcal{O}(\zeta T_{me-phe})$.
In terms of model detection, SRFed performs this process primarily on the server side using plaintext data, whereas ShieldFL requires multiple rounds of interaction to complete the encrypted model detection. This results in $\mathcal{O}(T_{md-defe}) \ll \mathcal{O}(T_{md-phe})$. Furthermore, SRFed enables the server to complete decryption and aggregation simultaneously. In contrast, ShieldFL necessitates aggregation prior to decryption and involves interactions with a third party, resulting in significantly higher overhead, i.e., $\mathcal{O}(T_{ma-defe}) \ll \mathcal{O}(T_{ma-phe})$.
Overall, these characteristics collectively render SRFed more efficient than ShieldFL.
The communication overhead of SRFed comprises two components: the encrypted models $\mathcal{O}(I\zeta|w_{defe}|)$ and key vectors $\mathcal{O}(IL|w|)$ uploaded by $I$ clients, where $\zeta$ denotes the model dimension, $L$ is the number of layers, $|w_{defe}|$ and $|w|$ are the communication complexity of a single DEFE ciphertext and a single plaintext, respectively. Since $|w|$ is significantly lower than $|w_{defe}|$, and $|w_{defe}|$ and $|w_{phe}|$ are nearly equivalent, SRFed reduces the overall communication complexity by approximately $\mathcal{O}(12I\zeta|w_{phe}|)$ compared to ShieldFL. This reduction in overhead is primarily attributed to SRFed’s lightweight DEFE scheme, which eliminates extensive third-party interactions.

\end{proof}

\begin{table}[tt]
	\centering
	\caption{Comparison of computation and communication overhead between different methods}\label{tab:com_comm}
    \resizebox{\linewidth}{!}{
	\begin{threeparttable}
		\renewcommand{\arraystretch}{1.4} 
		\begin{tabular}{clc}
			\toprule
			Method & \makecell[c]{SRFed}  &  \makecell[c]{ShieldFL}     \\ 
			\midrule
			\multirow{3}{*}[0pt]{Comp.}  
			& $\mathcal{O}(T_{lt}) + \mathcal{O}(\zeta T_{me-defe})$   & $\mathcal{O}(T_{lt}) + \mathcal{O}(\zeta T_{me-phe})$   \\
			& $+\mathcal{O}(T_{ma-defe})$ & $+\mathcal{O}(\zeta T_{md-phe})$ \\
			& $+ \mathcal{O}(T_{ma-defe})$ & $+ \mathcal{O}(T_{md-phe})$ \\ 
			\hline
			\multirow{1}{*}{Comm.}     
			& $\mathcal{O}(I\zeta|w_{defe}|)^1 + \mathcal{O}(IL|w|)^2 $  & $\mathcal{O}(13I\zeta|w_{phe}|)^{\mathrm{3}}$ \\
			\bottomrule
		\end{tabular}
		\begin{tablenotes}
			\footnotesize
			\item \textbf{Notes}: $^{\mathrm{1,2,3}} |w_{defe}|$, $|w|$ and $|w_{phe}|$ denote the communication complexity of a DEFE ciphertext, a plaintext, and a PHE ciphertext, respectively.
		\end{tablenotes}
	\end{threeparttable}}
\end{table}

\section{Experiments}
\subsection{Experimental Settings}
\subsubsection{Implementation} 

We implement SRFed on a small-scale local network. Each machine in the network is equipped with the following hardware configuration: an Intel Xeon CPU E5-1650 v4, 32 GB of RAM, an NVIDIA GeForce GTX 1080 Ti graphics card, and a network bandwidth of 40 Mbps. Additionally, the implementation of the DEFE scheme is based on the NDD-FE scheme \cite{NDDFE_2022_Big_data}, and the code implementation of FL processes is referenced to \cite{coding}.

\subsubsection{Dataset and Models}
We evaluate the performance of SRFed on two datasets:
\begin{itemize}
		\item MNIST \cite{mnist_database}: This dataset consists of 10 classes of handwritten digit images, with 60,000 training samples and 10,000 test samples. Each sample is a grayscale image of 28 × 28 pixels. The global model used for this dataset is a Convolutional Neural Network (CNN) model, which includes two convolutional layers followed by two fully connected layers.

        \item CIFAR-10 \cite{CiFar10}: This dataset contains RGB color images across 10 categories, including airplane, car, bird, cat, deer, dog, frog, horse, boat, and truck. It consists of 50,000 training images and 10,000 test samples. Each sample is a 32 × 32 pixel color image. The global model used for this dataset is a CNN model, which includes three convolutional layers, one pooling layer, and two fully connected layers. 

\end{itemize}

\subsubsection{Baselines}\label{baselines}

To evaluate the robustness of the proposed SRFed method, we conduct comparative experiments against several advanced baseline methods, including FedAvg \cite{FL}, ShieldFL \cite{HE_shieldfl}, PBFL \cite{pbfl}, Median \cite{Median_original}, Biscotti \cite{Mkrum}, and FoolsGold \cite{foolsgold}. Furthermore, to evaluate the efficiency of SRFed, we compare it with representative methods such as ShieldFL \cite{HE_shieldfl} and ESB-FL \cite{NDDFE_2022_Big_data}.

\subsubsection{Experimental parameters} 


In all experiments, the number of local clients is set to 20, the number of training rounds is set to 100, the batchsize is set to 64, and the number of local training epochs is set to 10. We use the stochastic gradient descent (SGD) to optimize the model, with a learning rate of 0.01 and a momentum of 0.5. Additionally, our experiments are conducted under varying levels of data heterogeneity, with the data distributions configured as follows:
\begin{itemize}
\item MNIST: Two distinct levels of data heterogeneity are configured by sampling from a Dirichlet distribution with the parameters $\alpha = 0.2$ and $\alpha = 0.8$, respectively, to simulate Non-IID data partitions across clients.
\item CIFAR-10: Two distinct levels of data heterogeneity are configured by sampling from a Dirichlet distribution with the parameters $\alpha = 0.2$ and $\alpha = 0.6$, respectively, to simulate Non-IID data partitions across clients.
\end{itemize}

\subsubsection{Attack Scenario}

In each benchmark, the adversary can control a certain proportion of clients to launch poisoning attacks, with the proportion varying across \{0\%, 10\%, 20\%, 30\%, 40\%, 50\%\}. The attack scenario parameters are configured as follows:
\begin{itemize}
\item Targeted Poisoning Attack: we consider the mainstream label-flipping attack. For experiments on the MNIST dataset, the training samples originally labeled as "0" are reassigned to the target label "4". For the CIFAR-10 dataset, the training samples originally labeled as "airplane" are reassigned to the target label "deer".
\item Untargeted Poisoning Attack: We consider the commonly used Gaussian attack. In experiments, malicious clients inject noise that follows a Gaussian distribution \(\mathcal{N}(0, 0.5^2)\) into their local model updates.
\end{itemize}

\subsubsection{Evaluation Metrics}

For each benchmark experiment, we adopt the following evaluation metrics on the test dataset to quantify the impact of poisoning attacks on the aggregated model in FL.
\begin{itemize}
\item Overall Accuracy (OA): It is the ratio of the number of samples correctly predicted by the model in the test dataset to the total number of predictions for all samples in the test dataset. 
\item Source Accuracy (SA): It specifically refers to the ratio of the number of correctly predicted flip class samples by the model to the total number of flip class samples in the dataset.
\item Attack Success Rate (ASR): It is defined as the proportion of source-class samples that are misclassified as the target class by the aggregated model.
\end{itemize}

\begin{figure}[!t]  
    \centering  
    \subfloat[MNIST ($\alpha$=0.2)]{\includegraphics[width=0.45\linewidth]{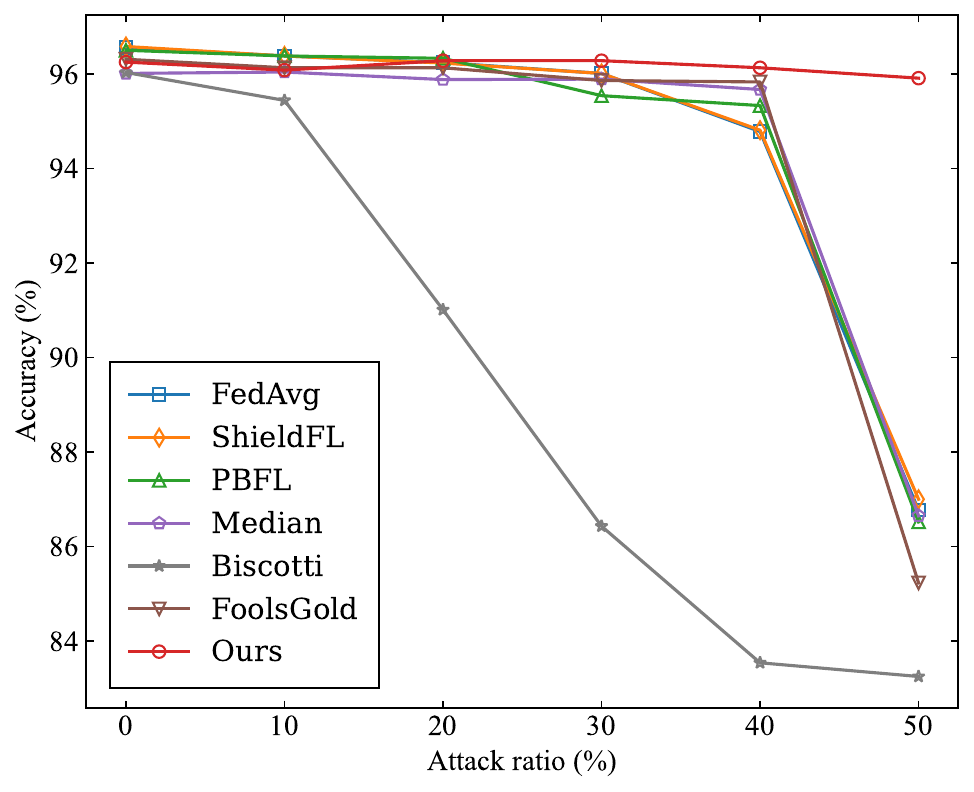}}
    \hfill  
    \subfloat[MNIST ($\alpha$=0.8)]{\includegraphics[width=0.45\linewidth]{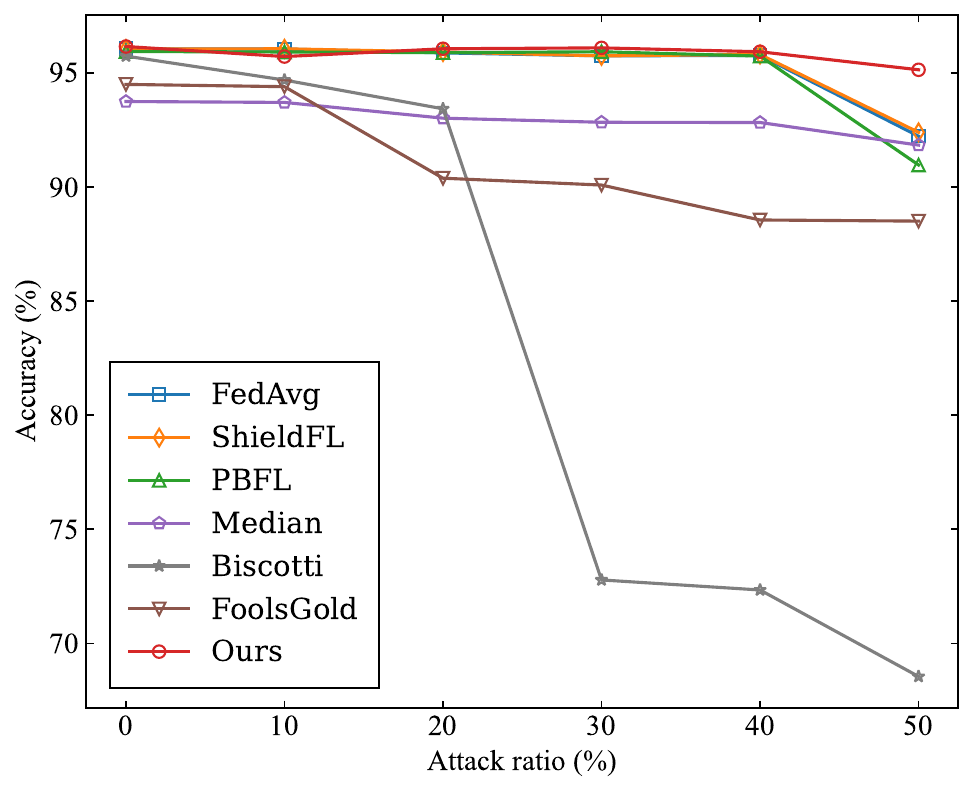}} \\
    \subfloat[CIFAR10 ($\alpha$=0.2)]{\includegraphics[width=0.45\linewidth]{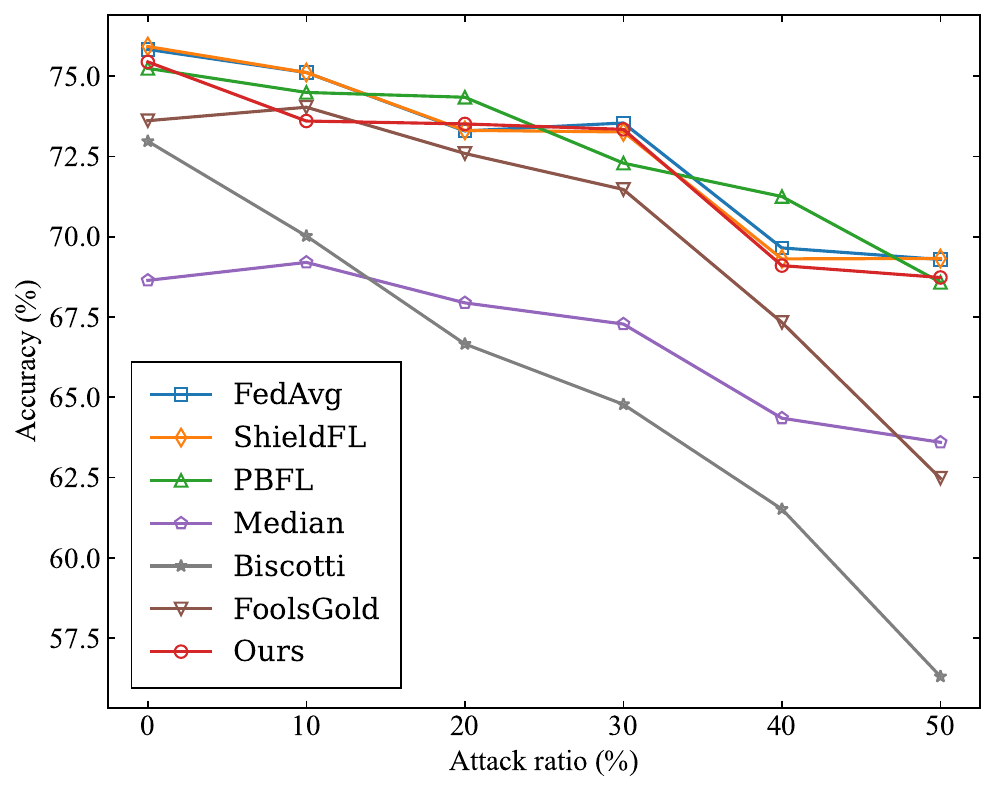}}
    \hfill
    \subfloat[CIFAR10 ($\alpha$=0.6)]{\includegraphics[width=0.45\linewidth]{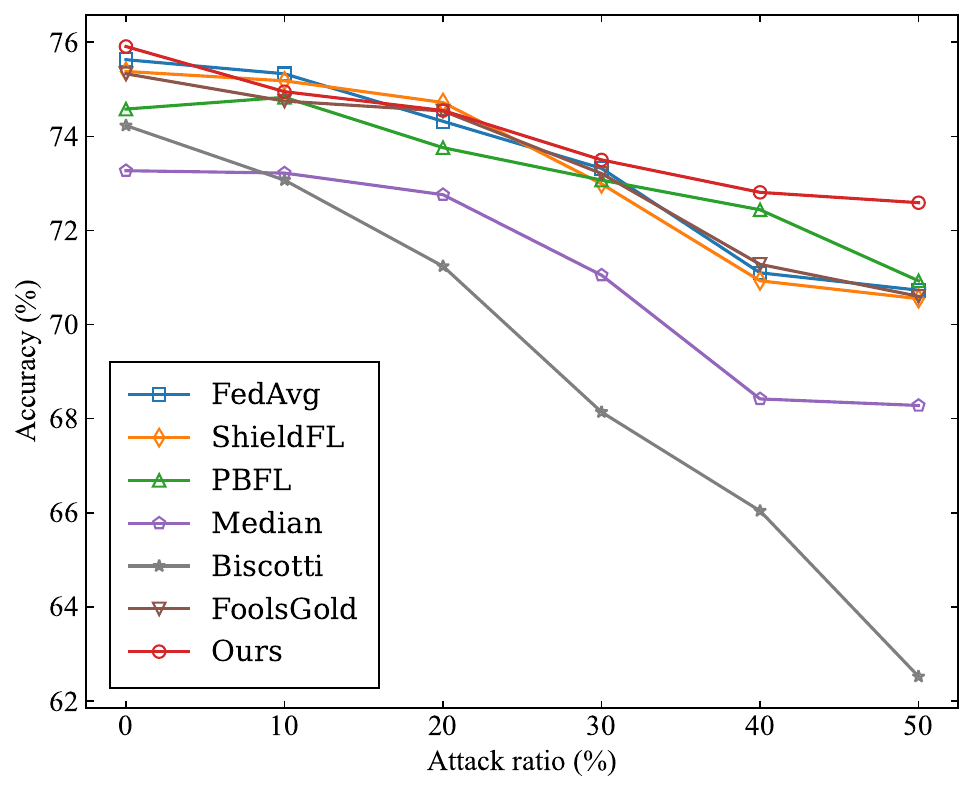}}
    \caption{The OA of the models obtained by four benchmarks under label-flipping attack.}  
    \label{lf_oa_benchmarks}  
\end{figure}
\begin{figure}[!t]  
    \centering  
    \subfloat[MNIST ($\alpha$=0.2)]{\includegraphics[width=0.45\linewidth]{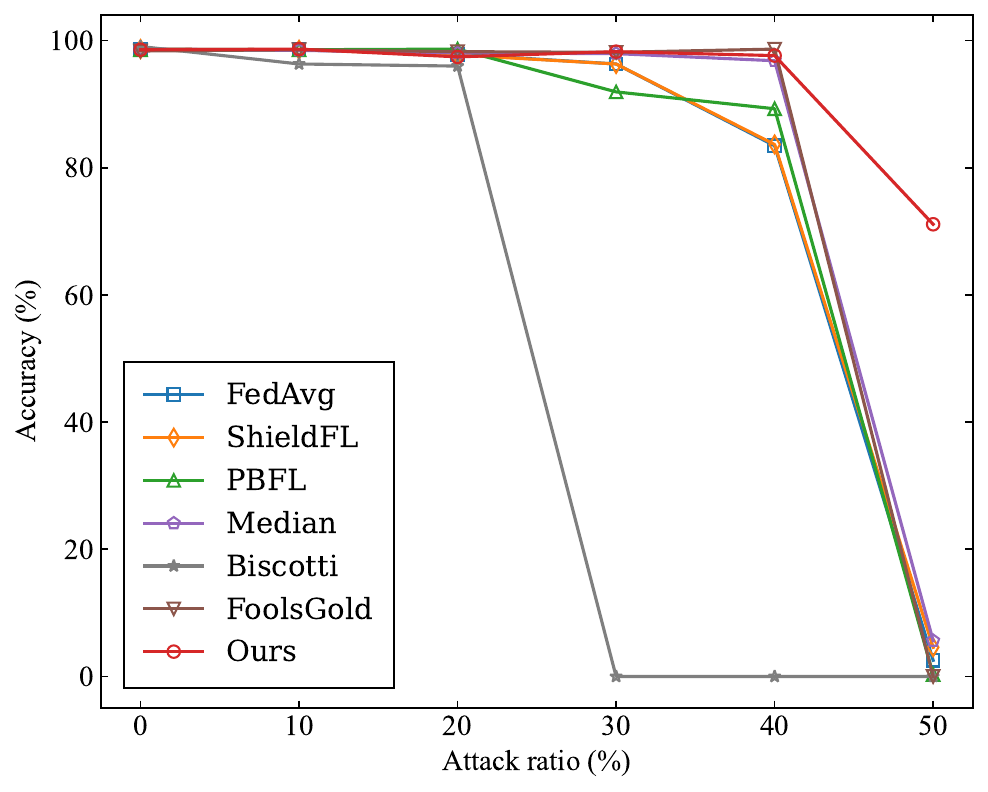}}
    \hfill  
    \subfloat[MNIST ($\alpha$=0.8)]{\includegraphics[width=0.45\linewidth]{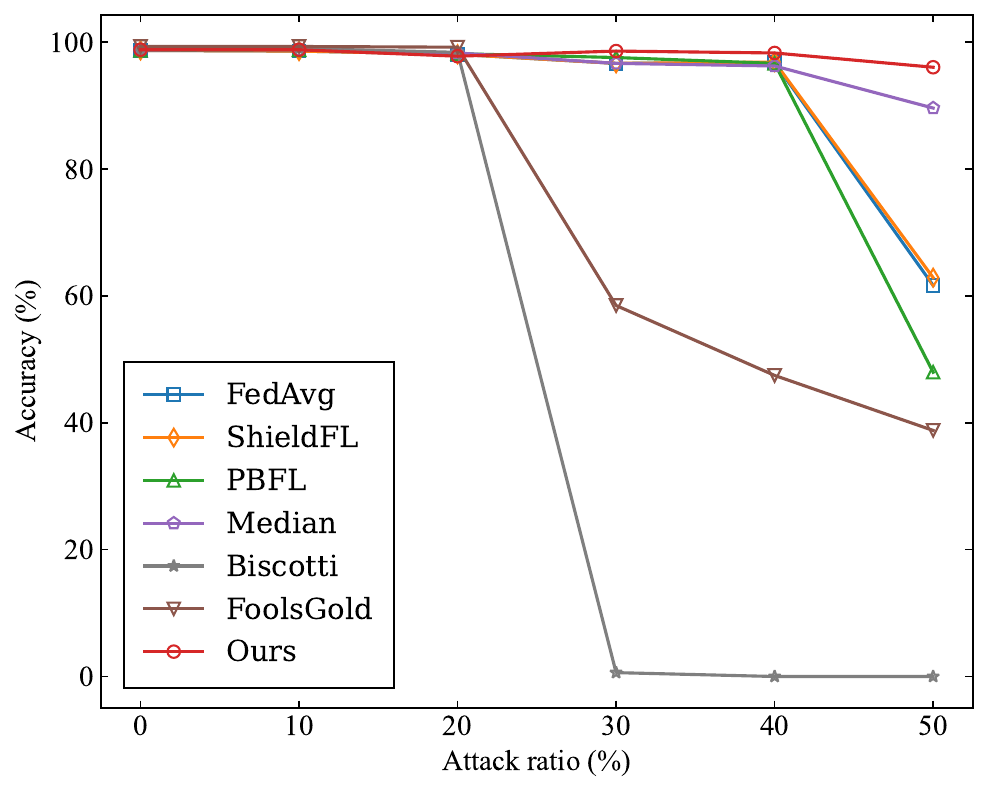}} \\
    \subfloat[CIFAR10 ($\alpha$=0.2)]{\includegraphics[width=0.45\linewidth]{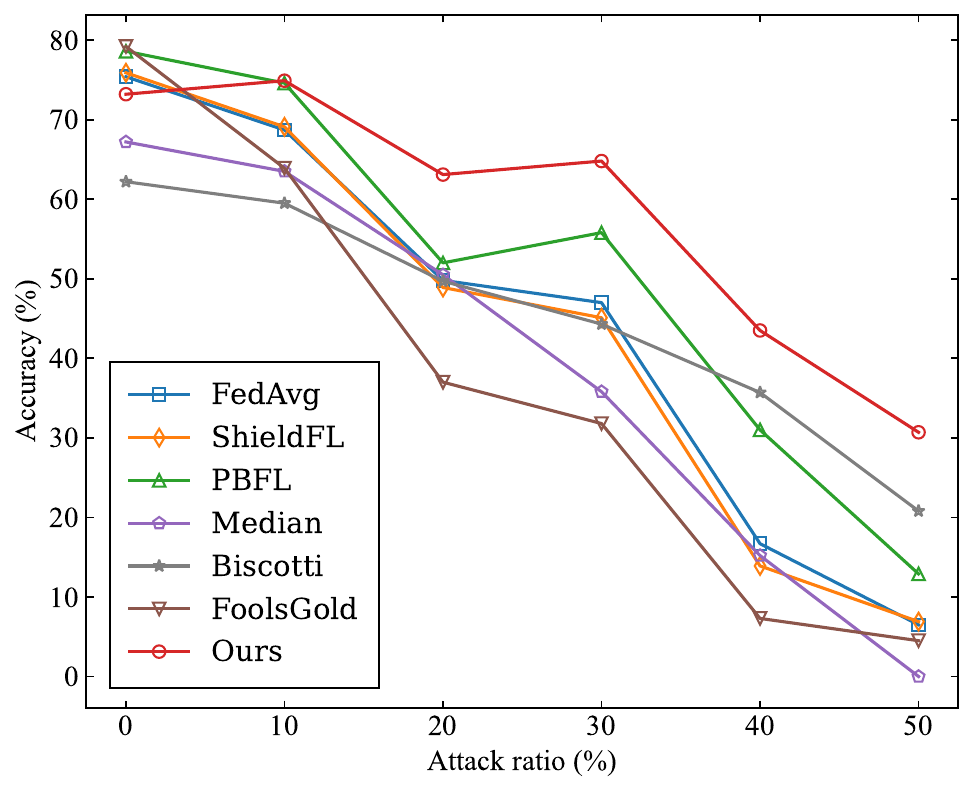}}
    \hfill
    \subfloat[CIFAR10 ($\alpha$=0.6)]{\includegraphics[width=0.45\linewidth]{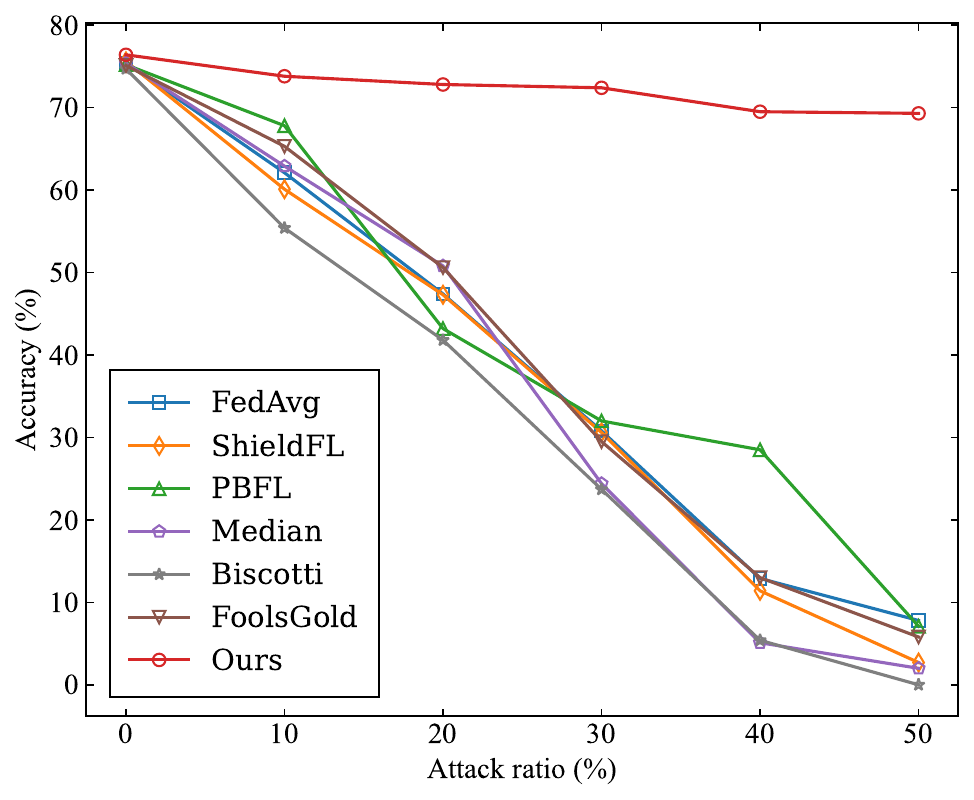}}
    \caption{The SA of the models obtained by four benchmarks under label-flipping attack.}  
    \label{lf_sa_benchmarks}  
\end{figure}

\subsection{Experimental Results}
\subsubsection{Robustness Evaluation of SRFed}

To evaluate the robustness of the proposed SRFed framework, we conduct a comparative analysis against the six baseline methods discussed in Section \ref{baselines}. Specifically, we first evaluate the overall accuracy (OA) of FedAvg, ShieldFL, PBFL, Median, Biscotti, FoolsGold, and SRFed under the label-flipping attacks. The results are presented in Figure \ref{lf_oa_benchmarks}. In the MNIST benchmarks with two levels of heterogeneity, the proposed SRFed consistently maintains a high OA across varying proportions of malicious clients. In the MNIST ($\alpha$=0.8) benchmark, all methods except Biscotti demonstrate relatively strong defense performance. Similarly, in the MNIST ($\alpha$=0.2) benchmark, all methods, except Biscotti, continue to perform well when facing a malicious client proportion ranging from 0\% to 40\%. Biscotti’s poor performance is due to the fact that, in Non-IID data scenarios, the model distributions trained by benign clients are more scattered, which can lead to their incorrect elimination. In the CIFAR-10 dataset benchmarks with different levels of heterogeneity, the OA of all methods fluctuates as the proportion of malicious clients increases. However, SRFed generally maintains better performance compared to the other methods, owing to the effectiveness of its robust aggregation strategy.

We further compare the SA of the global models achieved by different methods in the four benchmarks under label-flipping attacks. SA accurately measures the defense effectiveness of different methods against poisoning attacks, as it specifically reveals the model’s accuracy on the samples of the flipping label. The experimental results are presented in Figure \ref{lf_sa_benchmarks}. In the MNIST ($\alpha=0.2$) benchmark, SRFed demonstrates a significant advantage in defending against label-flipping attacks. Especially, when the attack ratio reaches 50\%, SRFed achieves a SA of 70\%, while the SA of all other methods drops to nearly 0\% even though their OA remaining above 80\%. SRFed is the only method to sustain a high SA across all attack ratios, underscoring its superior Byzantine robustness even in scenarios with extreme data heterogeneity and high attack ratios. In the MNIST ($\alpha=0.8$) benchmark, SRFed also outperforms other baselines. In the CIFAR-10 ($\alpha=0.2$) benchmark, although SRFed still outperforms the other methods, its performance gradually deteriorates as the proportion of malicious clients increases. This demonstrates that defending against poisoning attacks in scenarios with a high attack ratio and extremely heterogeneous data remains a significant challenge. In the CIFAR-10 ($\alpha=0.6$) benchmark, SRFed maintains a high level of performance as the proportion of malicious clients increases (SA $\ge$ 70\%), while the SA of all other methods sharply declines and eventually approaches 0\%. This superior performance is attributed to the robust aggregation strategy of SRFed, which performs layer-wise projection and clustering analysis on client models. This enables more accurate detection of local parameter anomalies compared to baselines.

We also evaluate the ASR of the models obtained by different methods across four benchmarks, with the experimental results presented in Figure \ref{lf_asr_benchmarks}. As the attack ratio increases, we can observe that the ASR trend exhibits a negative correlation with the SA trend. Notably, our proposed SRFed consistently demonstrates optimal performance across all four benchmarks, showing minimal performance fluctuations across varying attack ratios.

\begin{figure}[!t]  
    \centering  
    \subfloat[MNIST ($\alpha$=0.2)]{\includegraphics[width=0.45\linewidth]{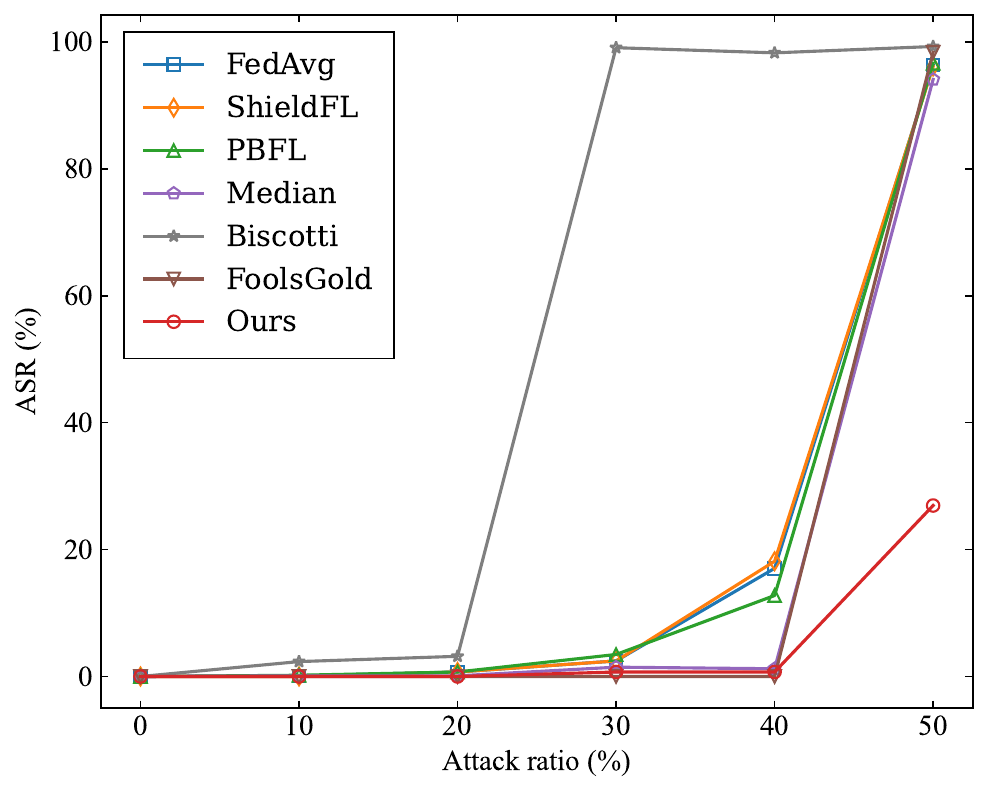}}
    \hfill  
    \subfloat[MNIST ($\alpha$=0.8)]{\includegraphics[width=0.45\linewidth]{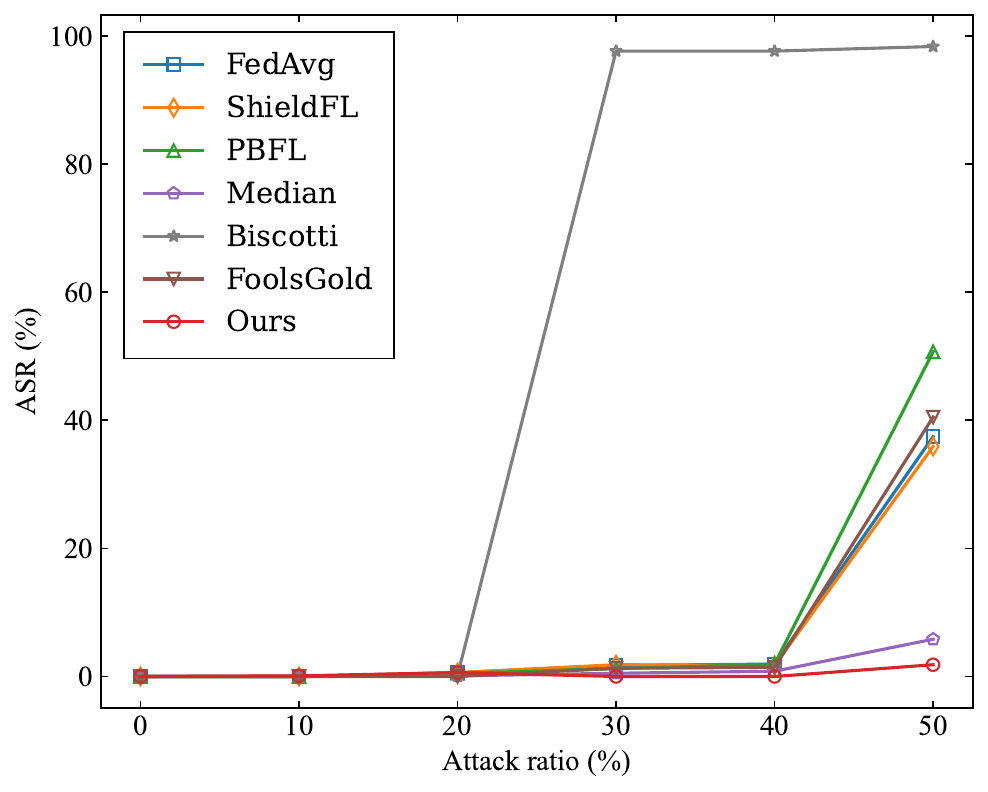}} \\
    \subfloat[CIFAR10 ($\alpha$=0.2)]{\includegraphics[width=0.45\linewidth]{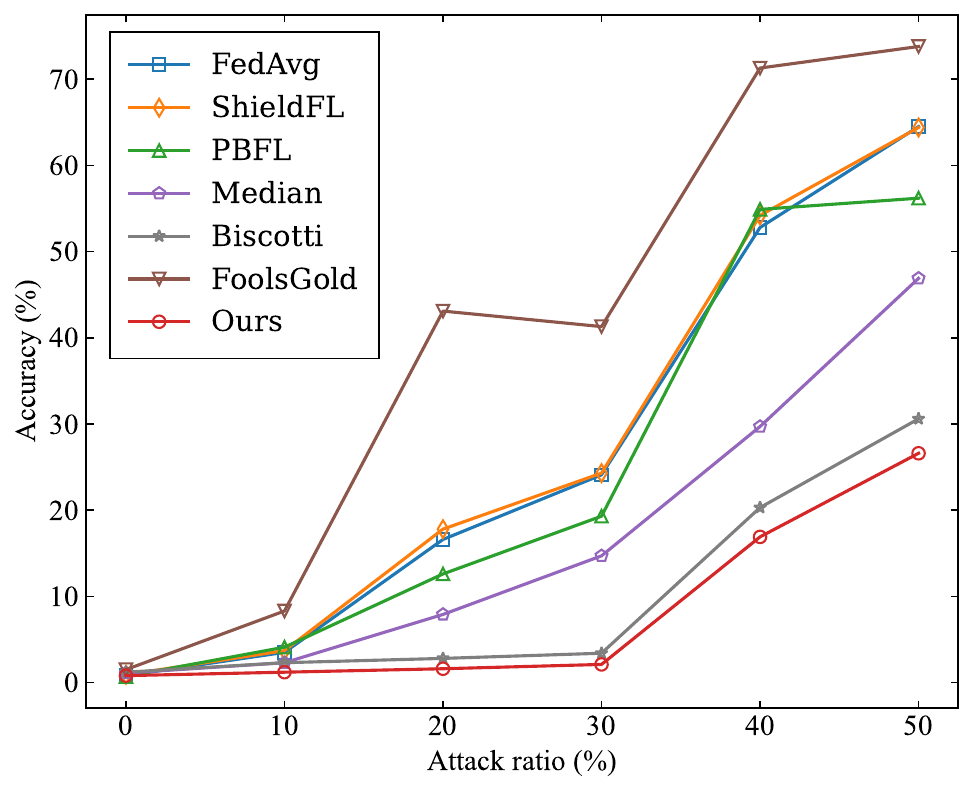}}
    \hfill
    \subfloat[CIFAR10 ($\alpha$=0.6)]{\includegraphics[width=0.45\linewidth]{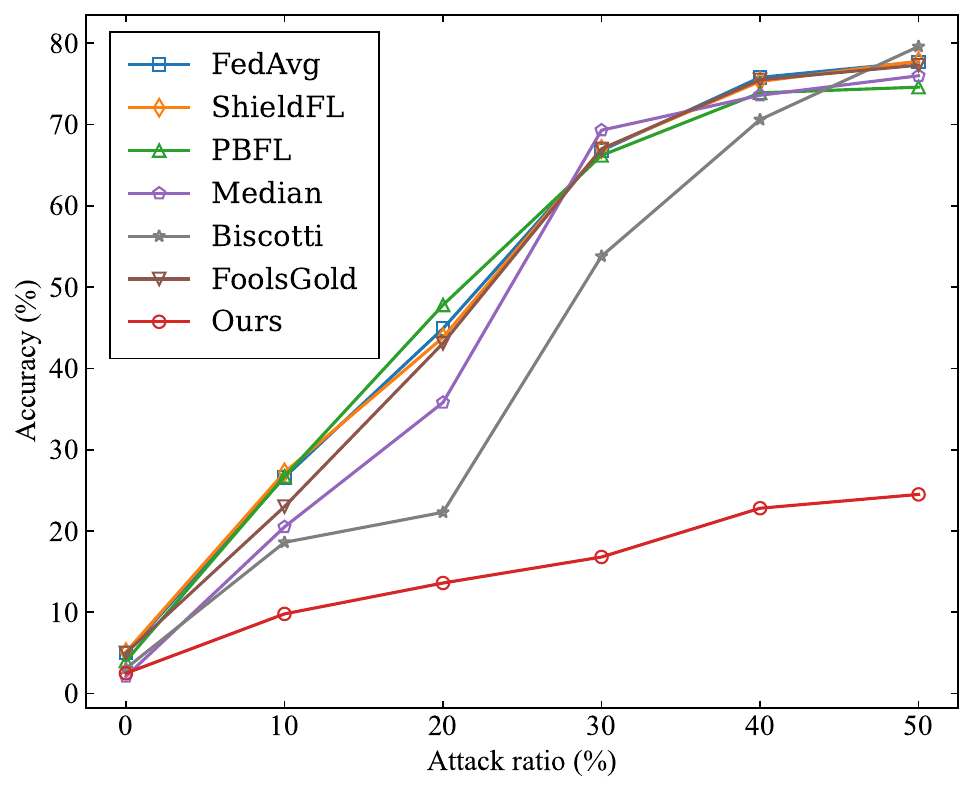}}
    \caption{The ASR of the models obtained by four benchmarks under label-flipping attack.}  
    \label{lf_asr_benchmarks}  
\end{figure}

Finally, we evaluate the OA of the models of different methods under the Gaussian attack. The experimental results are shown in Figure \ref{Ga_oa_benchmarks}. We observe that SRFed consistently achieves optimal performance across all four benchmarks. Furthermore, as the attack ratio increases, SRFed exhibits minimal fluctuations in OA. Specifically, in the MNIST ($\alpha=0.2$) and MNIST ($\alpha=0.8$) benchmarks, all methods maintain an OA above 90\% when the attack ratio is $\leq$ 20\%. However, when the attack ratio $\ge$ 30\%, only SRFed and Median retain an OA above 90\%, demonstrating their effective defense against poisoning attacks under high malicious client ratios. In the CIFAR-10 ($\alpha=0.2$) and CIFAR-10 ($\alpha=0.6$) benchmarks, while the OA of most methods drops below 30\% as the attack ratio increases, SRFed consistently maintains high accuracy across all attack rates, demonstrating its robustness against extreme client ratios and heterogeneous data distributions.

\begin{figure}[!t]  
    \centering  
    \subfloat[MNIST ($\alpha$=0.2)]{\includegraphics[width=0.45\linewidth]{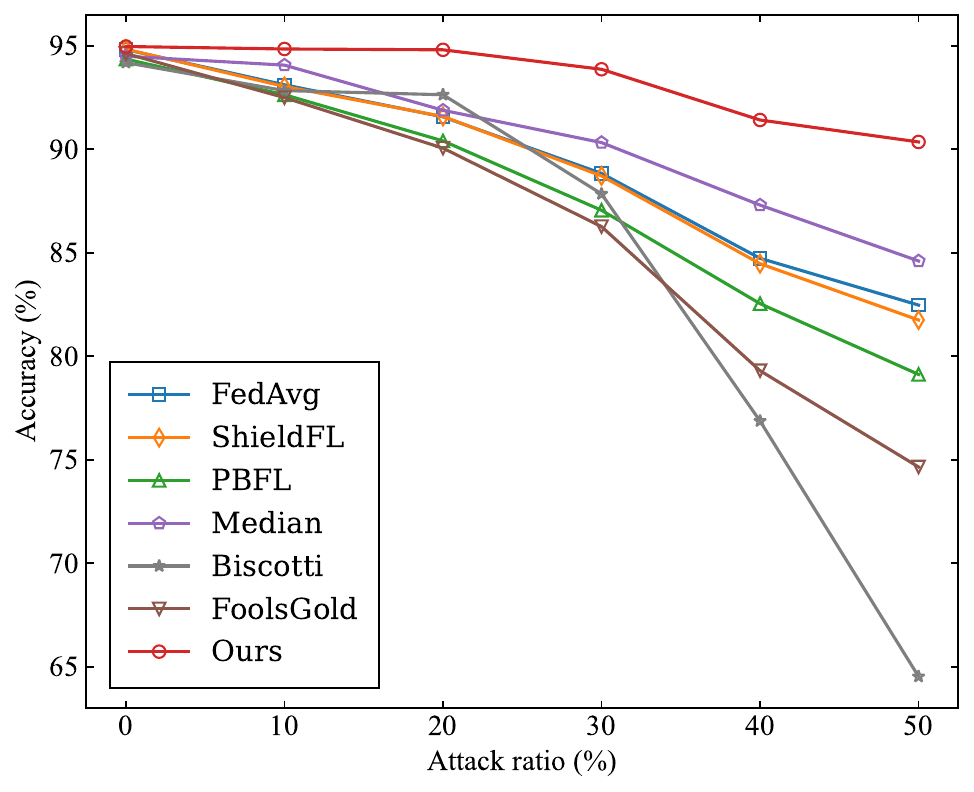}}
    \hfill  
    \subfloat[MNIST ($\alpha$=0.8)]{\includegraphics[width=0.45\linewidth]{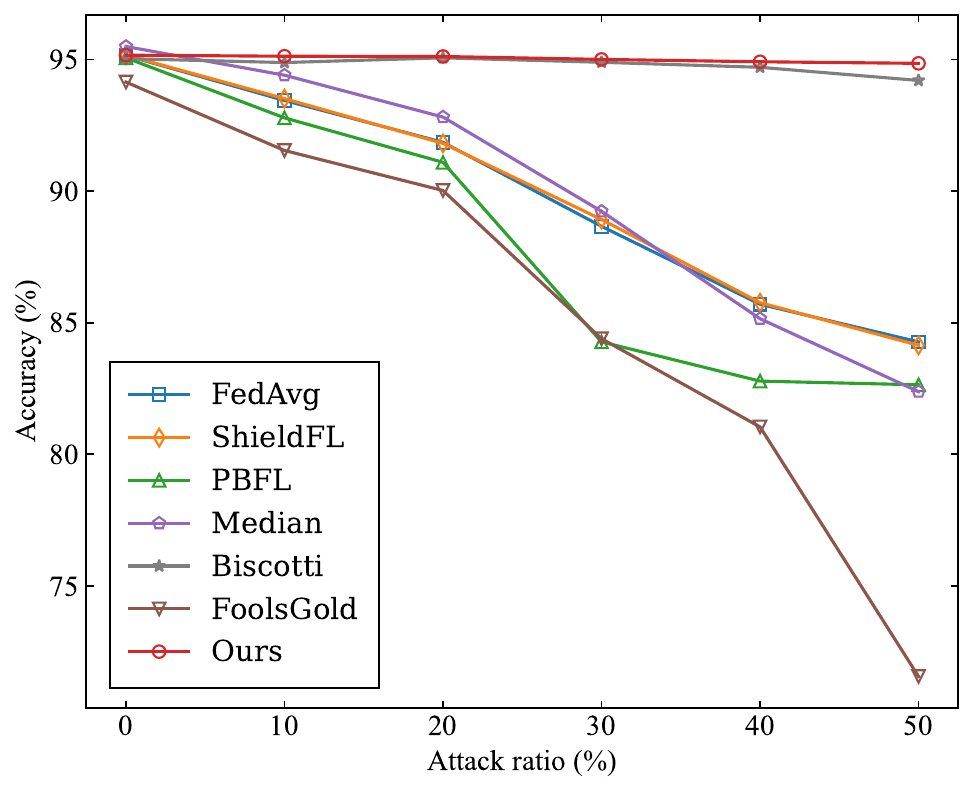}} \\
    \subfloat[CIFAR10 ($\alpha$=0.2)]{\includegraphics[width=0.45\linewidth]{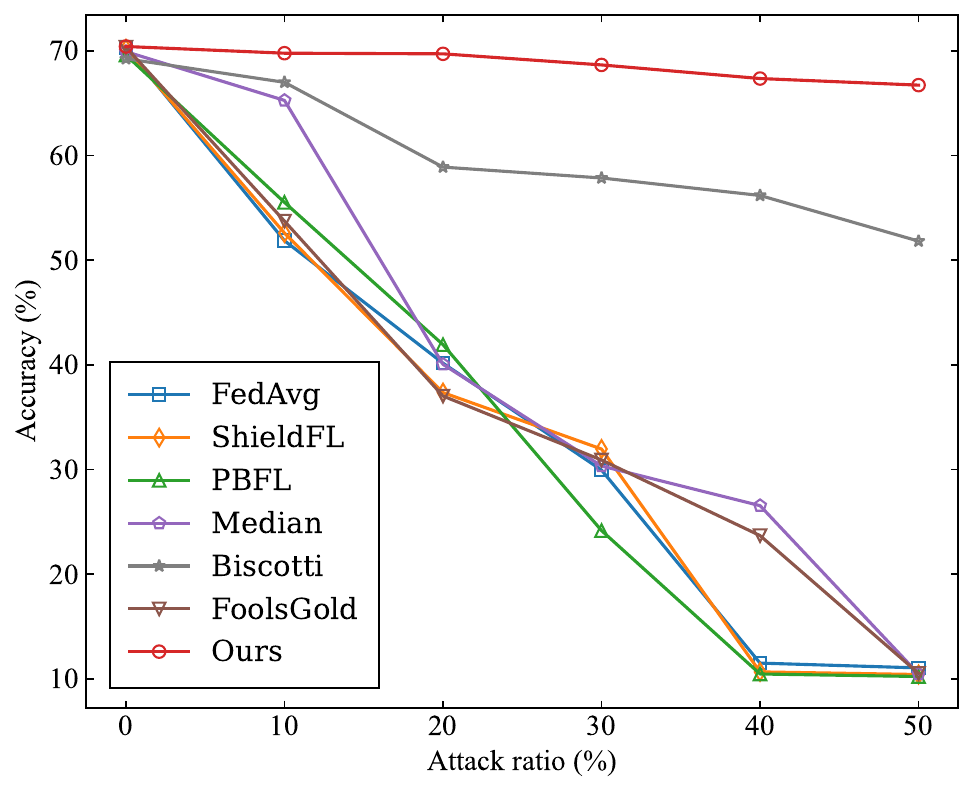}}
    \hfill
    \subfloat[CIFAR10 ($\alpha$=0.6)]{\includegraphics[width=0.45\linewidth]{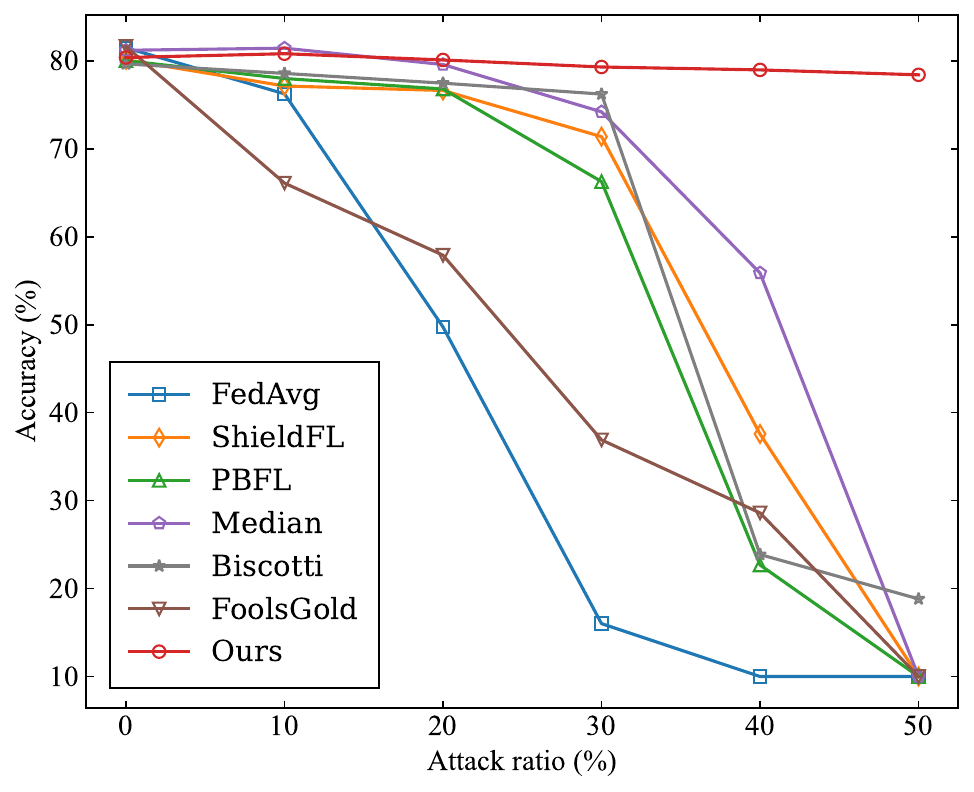}}
    \caption{The OA of the models obtained by four benchmarks under Gaussian attack.}  
    \label{Ga_oa_benchmarks}  
\end{figure}

In summary, SRFed demonstrates strong robustness against poisoning attacks under different Non-IID data settings and attack ratios, thus achieving the design goal of robustness.


\subsubsection{Efficiency Evaluation of SRFed} 
\textbf{Learning Overheads.} We evaluate the efficiency of the proposed SRFed in obtaining a qualified aggregated model. Specifically, we compare SRFed with two baseline methods, i.e., ESB-FL and ShieldFL. These two methods respectively utilize NDD-FE and HE to ensure privacy protection for local models. The experiments are conducted on MNIST with no malicious clients. For each method, we conduct 10 training tasks and calculate the average time consumed in each phase, along with the average communication time across all participants. The results are summarized in Table \ref{tab:time_consumption}. The experimental results demonstrate that SRFed reduces the total time overheads throughout the entire training process by 58\% compared to ShieldFL. This reduction can be attributed to two main factors: 1) DEFE in SRFed offers a significant computational efficiency advantage over HE in ShieldFL, with faster encryption and decryption, as shown in the "Local training" and "Privacy-preserving robust model aggregation" phases in Table \ref{tab:time_consumption}. 2) The privacy-preserving robust model aggregation is handled solely by the server, which avoids the overhead of multi-server interactions in ShieldFL. Compared to ESB-FL, SRFed reduces the total time overhead by 22\% even though it incorporates an additional privacy-preserving model detection phase. This is attributed to its underlying DEFE scheme, which significantly enhances decryption efficiency. As a result, SRFed achieves a 71\% reduction in execution time during the privacy-preserving robust model aggregation phase, even with the added overhead of model detection. In summary, SRFed achieves an efficient privacy-preserving FL process, achieving the design goal of efficiency.

\begin{table}[tt]
	\centering
	\caption{Comparison of time consumption between different frameworks}\label{tab:time_consumption}
    \resizebox{\linewidth}{!}{
	\begin{threeparttable}
		\renewcommand{\arraystretch}{1.4}
		\begin{tabular}{cccc}
			\toprule
			Framework & SRFed & ShieldFL & ESB-FL \\ 
			\midrule
			\multirow{1}{*}{Local training$^1$} 
			& 19.51 h & 14.23 h & 5.16 h \\
			
			Privacy-preserving   
			& \multirow{2}{*}{9.09 h} & \multirow{2}{*}{51.97 h} & \multirow{2}{*}{31.43 h} \\
             robust model aggregation$^2$ & & & \\
			
			\multirow{1}{*}{Node communication}  
			& 0.09 h & 1.51 h & 0.09 h \\
			
			\textbf{Total time}  
			& \textbf{28.69 h} & \textbf{67.71 h} & \textbf{36.68 h} \\
			
			\textbf{Accuracy}  
			& \textbf{98.90\%} & \textbf{97.42\%} & \textbf{98.68\%} \\
			\bottomrule
		\end{tabular}
	\end{threeparttable}
    }
\end{table}

\begin{table}[tt]
	\centering
	\caption{Time overhead of proposed DEFE}\label{tab:comp_overhead}
	\begin{threeparttable}
		\renewcommand{\arraystretch}{1.4}
		\begin{tabular}{lccc}
			\toprule
			Operations & DEFE & NDD-FE & HE \\ 
			\multicolumn{1}{l}{(for a model)} & (SRFed) & (ESB-FL) & (shieldFL) \\
			\midrule
			\multirow{1}{*}{Encryption}  
			& 28.37 s & 2.53 s & 18.87 s \\
			
			\multirow{1}{*}{Inner product}  
			& 8.97 s & 56.58 s & 30.15 s \\
			
			\multirow{1}{*}{Decryption}  
			& - & - & 3.10 s \\
			\bottomrule
		\end{tabular}
	\end{threeparttable}
\end{table}

\textbf{Efficiency Evaluation of DEFE.} We further evaluate the efficiency of the DEFE scheme within SRFed by conducting experiments on the CNN model of the MNIST dataset. Specifically, we compare the DEFE scheme with the NDD-FE scheme used in ESB-FL \cite{NDDFE_2022_Big_data} and the HE scheme used in ShieldFL \cite{HE_shieldfl}. For these schemes, we calculate their average time required for different operations, i.e., encryption, inner product computation, and decryption, over 100 test runs. The results are presented in Table \ref{tab:comp_overhead}. It is evident that the DEFE scheme offers a substantial efficiency advantage in terms of inner product computation. Specifically, the inner product time of DEFE is reduced by 84\% compared to NDD-FE and by 70\% compared to HE. Furthermore, DEFE directly produces the final plaintext result during inner product computation, avoiding the need for interactive decryption in HE. Combined with the results in Table \ref{tab:time_consumption}, it is clear that although the encryption time of DEFE is slightly higher, its highly efficient decryption process significantly reduces the overall computation overhead. Thus, DEFE guarantees the high efficiency of SRFed.

\section{Conclusion}

In this paper, we address the challenges of achieving both privacy preservation and Byzantine robustness in FL under Non-IID data distributions, and propose a novel secure and efficient FL method SRFed. Specifically, we design a DEFE scheme that enables efficient model encryption and non-interactive decryption, which eliminates third-party dependency and defends against server-side inference attacks. Second, we develop a privacy-preserving robust aggregation mechanism based on secure layer-wise projection and clustering, which effectively filters malicious updates and mitigates poisoning attacks in data heterogeneous environments. Theoretical analysis and extensive experimental results demonstrate that SRFed achieves superior performance compared to state-of-the-art baselines in terms of privacy protection, Byzantine resilience, and system efficiency. In future work, we will explore the extension of SRFed to practical FL scenarios, such as vertical FL, edge computing, and personalized FL.




\begin{wrapfigure}{l}{25mm}   
  \centering
  \includegraphics[width=1in,height=1.25in,clip,keepaspectratio]{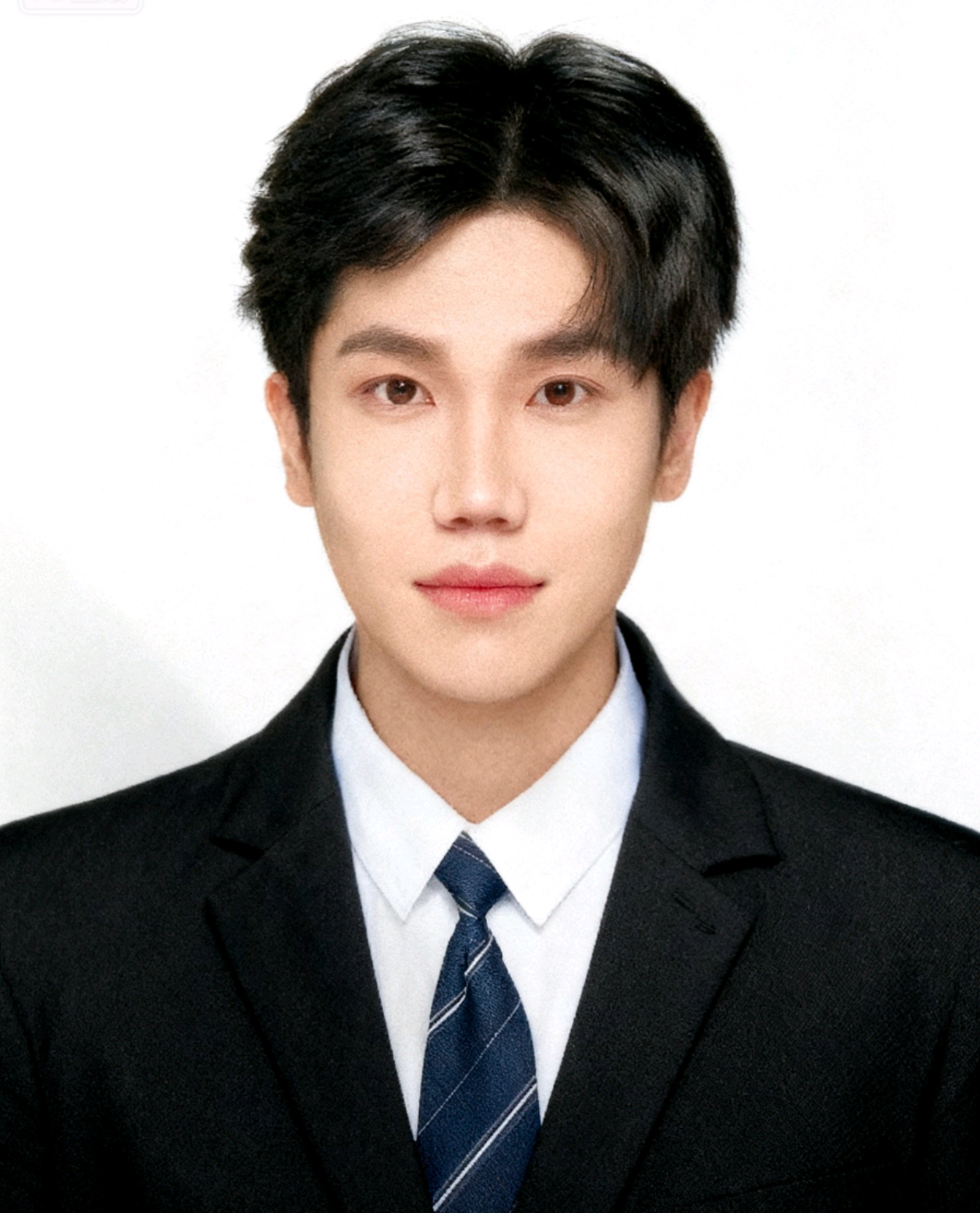}
\end{wrapfigure}

\vspace*{\baselineskip}
\noindent\textbf{Yiwen Lu} received the B.S. degree from the School of Mathematics and Statistics, Central South University (CSU), Changsha, China, in 2021. He is currently working toward the Ph.D. degree in mathematics with the School of Mathematics, Nanjing University (NJU), Nanjing, China. His research interests include number theory, cryptography, and artificial intelligence security.
\end{document}